%% file: main.tex
\documentclass[a4paper,UKenglish,cleveref,autoref,thm-restate,numberwithinsect]{lipics-v2021}
\pdfoutput=1 %uncomment to ensure pdflatex processing (mandatatory e.g. to submit to arXiv)
\hideLIPIcs  %uncomment to remove references to LIPIcs series (logo, DOI, ...), e.g. when preparing a pre-final version to be uploaded to arXiv or another public repository

\title{The Flower Calculus}

\author{Pablo Donato}{LIX, École polytechnique, France}{pablo.donato@polytechnique.edu}{https://orcid.org/0000-0001-7883-6754}{}

\authorrunning{P.~Donato}

\Copyright{Pablo Donato}



\category{} %optional, e.g. invited paper

\relatedversion{} %optional, e.g. full version hosted on arXiv, HAL, or other respository/website
\supplement{The source code for a Coq mechanization of additional meta-theoretical results, as well as a web demo of \kl{GUI} for \kl{ITPs} based on this work, are available as follows:}%optional, e.g. related research data, source code, ... hosted on a repository like zenodo, figshare, GitHub, ...
\supplementdetails[cite=flowers-metatheory, subcategory={Mechanized Theory}]{Software}{https://github.com/Champitoad/flowers-metatheory} %linktext, cite, and subcategory are optional
\supplementdetails[cite=flower-prover, subcategory={Online Demo}]{Software}{https://github.com/Champitoad/flower-prover} %linktext, cite, and subcategory are optional

\acknowledgements{I want to thank Luc Chabassier for writing the Lua script that
was used to generate all the flower drawings in this document, and Benjamin
Werner for useful feedback on a first draft of this paper. Lastly, I thank Tito
for teaching me some invaluable formatting tricks.}%optional

\nolinenumbers %uncomment to disable line numbering

\EventEditors{John Q. Open and Joan R. Access}
\EventNoEds{2}
\EventLongTitle{42nd Conference on Very Important Topics (CVIT 2016)}
\EventShortTitle{CVIT 2016}
\EventAcronym{CVIT}
\EventYear{2016}
\EventDate{December 24--27, 2016}
\EventLocation{Little Whinging, United Kingdom}
\EventLogo{}
\SeriesVolume{42}
\ArticleNo{23}
\input{preamble.tex}

\begin{document}

\maketitle

\begin{abstract}
We introduce the flower calculus, a deep inference proof system for
intuitionistic first-order logic inspired by Peirce's existential graphs. It
works as a rewriting system over inductive objects called ``flowers'', that
enjoy both a graphical interpretation as topological diagrams, and a textual
presentation as nested sequents akin to coherent formulas. Importantly, the
calculus dispenses completely with the traditional notion of symbolic
connective, operating solely on nested flowers containing atomic predicates. We
prove both the soundness of the full calculus and the completeness of an
analytic fragment with respect to Kripke semantics. This provides to our
knowledge the first analyticity result for a proof system based on existential
graphs, adapting semantic cut-elimination techniques to a deep inference
setting. Furthermore, the kernel of rules targetted by completeness is fully
invertible, a desirable property for both automated and interactive proof
search.

% We introduce the \kl{flower calculus}, a novel proof system for intuitionistic
% first-order logic based on syntactic objects called flowers. We start by
% explaining how flowers stem from considerations in graphical logic, and more
% specifically from a constructivization of the existential graphs of C. S.
% Peirce proposed by A. Oostra. Then we present our inductive syntax for
% flowers, reminiscent at the same time of the nested sequents of deep inference
% proof theory, and the coherent formulas of categorical logic. A salient
% feature of our calculus is that it dispenses completely with the traditional
% notions of symbolic formula and proof tree: instead it operates as a rewriting
% system on flowers containing only atomic predicates. We also propose a notion
% of proof geared towards analyticity results à la Gentzen, suggesting new rules
% absent from other works on existential graphs. This allows us to prove
% admissibility theorems for many structural rules, including Peirce's erasure
% rule used in simulating the cut rule. These results are obtained as a
% consequence of our soundness and completeness proofs with respect to Kripke
% semantics, in the spirit of normalization-by-evaluation. Furthermore, the
% kernel of rules targetted by completeness is fully invertible, a desirable
% property in both automatic and interactive proof search.
\end{abstract}

\section{Introduction}

% This work stems from an investigation into \emph{diagrammatic} systems for
% logical reasoning, with an eye toward the development of graphical user
% interfaces (GUIs) for interactive theorem provers (\kl{ITPs}). Our main goal was to
% find a sound and complete proof system for first-order logic (FOL), whose rules
% could directly map to \emph{gestural} actions performed upon the representation
% of logical statements in the interface. Furthermore, the system should support
% reasoning in \emph{intuitionistic} logic, to be compatible with the trend of
% constructive type theories in modern proof assistants such as Coq
% \cite{the_coq_development_team_2022_7313584} and Agda \cite{agda}.

% A first step in this direction was achieved in \cite{10.1145/3497775.3503692}
% and \cite{DBLP:conf/cade/Chaudhuri21}, where the \emph{subformula linking}
% technique --- introduced by Chaudhuri in \cite{Chaudhuri2013} --- was adapted to
% intuitionistic FOL, and implemented in prototypes of GUIs for \kl{ITPs}
% (\emph{Actema} \cite{Actema:link} and ProfInt \cite{ProfInt}, respectively). In
% these systems, complete proofs can be built by the sole act of designating two
% subformulas of the current goal, using pointing interactions like \emph{click}
% and \emph{drag-and-drop} until the empty goal is reached. But fundamentally, the
% subformula linking technique is based on \emph{symbolic} means. We wanted to
% avoid the use of symbolic connectives, in order to make the system more
% accessible to users unfamiliar with formal logic.

\subparagraph*{Graphical proof building}

% Given a sufficiently expressive logic, the problem of finding a formal proof of
% a given statement in that logic is notoriously hard. Not only because of
% undecidability, but because proof calculi are very rigid, and expose too many
% details compared to the informal proofs that mathematicians are used to write.

% One way to tackle this limitation is to design search procedures for interesting
% fragments of the logic under consideration. This is the approach undertaken by
% automated theorem provers such as SMT solvers, or more specialized decision
% procedures for particular theories like linear integer arithmetic. Still,
% large-scale proofs of complex theorems currently require a human in the loop to
% provide the overall structure of arguments, but also of the theories built up
% from definitions and lemmas.

\AP
Proof assistants --- also called \intro{interactive theorem provers}
(\reintro{ITPs}) --- provide a set of tools to ease the process of formalizing
mathematical developments. This includes languages to specify definitions and
statements conveniently, but also interfaces to build proofs interactively
without having to fill in all the details. The dominant paradigm for these
interfaces is that of \emph{tactic languages} \cite{doi:10.1098/rsta.1984.0067}:
the user is exposed with a set of \intro{goals} that remain to be proved,
constituting the \intro{proof state}, and modifies these goals through textual
commands, called \intro{tactics}, until there is no goal left. This is currently
what is implemented in mainstream proof assistants such as Coq
\cite{the_coq_development_team_2022_7313584} and Lean
\cite{10.1007/978-3-030-79876-5_37}.

% \subparagraph*{Graphical interfaces}

\AP
In recent years, there have been several efforts to replace or complement
textual tactic languages with \intro{graphical user interfaces} (\reintro{GUIs})
\cite{ritchie_design_1988,PbP,langbacka-tkwinhol-1995,Chaudhuri2013,edukera,linker_tactical_2017,zhan-design-2019,ayers_graphical_2021}.
The hope is to make proof assistants more intuitive and accessible to beginners
and non-specialists, but also, to some extent, more productive and ergonomic
even for experts. 

\AP
The initial motivation for this work was to design a proof calculus well-suited
to \emph{direct manipulation} in such a graphical setting. The idea is that the
user should be able to interact directly with the graphical representation of
the \kl{proof state}, using a pointing device such as a mouse or fingers on a
touch screen. In previous work \cite{10.1145/3497775.3503692}, we proposed a way
to synthesize complex logical inferences through \emph{drag-and-drop} actions
between two formulas of the current goal/sequent, based on the \intro{subformula linking}
(\reintro{SFL}) methodology \cite{Chaudhuri2013,DBLP:conf/cade/Chaudhuri21}.
% Since goals are represented as \emph{sequents} $~\Gamma \seq C~$ in most
% \kl{ITPs}, the items involved were traditional logical formulas, either two
% hypotheses in $\Gamma$ or one hypothesis and the conclusion $C$.

\subparagraph*{Diagrammatic reasoning}

\AP
In this work, we show that (single-conclusion) sequents and symbolic formulas
built from binary connectives and unary quantifiers are not mandatory for
representing the \kl{proof state}. Other authors have defended the idea of using
\emph{diagrams} as a more user-friendly frontend for \kl{ITPs}. In particular,
Linker et al.\ showed how to integrate tactic-based automation in an \kl{ITP}
based on \emph{spider diagrams} \cite{linker_tactical_2017}, which are
equivalent in expressive power to classical monadic \intro{first-order logic}
(\reintro{FOL}) \cite{howse_spider_2005}.

% We speculate that they actually get in the way of a natural flow of
% information in the proof, and do not capture accurately the way mathematicians
% organise hypotheses and conclusions in their reasoning. This has already been
% E. Ayers in his thesis \cite[Chapter 3]{ayers_thesis}, where his solution is
% to design a nested goal data structure called \texttt{Box} on which reasoning
% can be performed directly, reducing the need to manipulate symbolic formulas.
\AP
We introduce a new data structure for \kl{goals} inspired by an earlier
invention in the history of diagrammatic logic: the \intro{existential graphs}
(\reintro{EGs}) of C. S. Peirce \cite{Roberts+1973}. We noticed that our structure
could be drawn and manipulated metaphorically in the form of nested
\emph{flowers}, and thus chose to name \intro{flower calculus} the proof system
for full intuitionistic \kl{FOL} that we built around it. Our focus in this paper
will be to introduce the \kl{flower calculus} to readers unfamiliar with
\kl{EGs}, and to study its fundamental properties through the lens of modern
\emph{structural proof theory}.

\subparagraph*{Implementation}

We have formalized in Coq a bidirectional simulation between the \kl{flower
calculus} and cut-free sequent calculus, yielding a soundness theorem and a
\emph{weak} completeness theorem for an \kl{analytic} fragment of the \kl{flower
calculus} \cite{flowers-metatheory}. In this paper, we follow a \emph{semantic}
rather than syntactic approach, avoiding translations to and from symbolic
formulas to obtain a stronger completeness result.

% « very » en moins : économise 1 ligne
\AP
While currently at an early stage, we are also developing the \intro{Flower
Prover}, a prototype of direct-manipulation GUI for \kl{ITPs} based on the
\kl{flower calculus} \cite{flower-prover}. The interested reader can try a
publicly available version of the prototype
online\footnote{\url{https://www.lix.polytechnique.fr/Labo/Pablo.DONATO/flowerprover/}}.
We leave a detailed account of the \kl{Flower Prover} and its connection to the
\kl{flower calculus} for future work.

\subparagraph*{Outline}

The article is organized as follows: in Section~\ref{sec:EGs} we give a brief
overview of the original diagrammatic syntax of \kl{EGs} used by Peirce in his
system \kl{Alpha} for classical propositional logic. In
Section~\ref{sec:Flowers} we retrace the origin of an intuitionistic variant of
\kl{EGs} first introduced by Oostra in \cite{oostra_graficos_2010}, that
directly inspired our flower metaphor. In Section~\ref{sec:Quantifiers} we
illustrate quickly the original mechanism of \kl{lines of identity} used by
Peirce to express first-order quantifiers in his \kl{Beta} system, and show how
to recast it in a more traditional binder-based syntax. In
Section~\ref{sec:Syntax} we introduce our inductive syntax for \kl{flowers}, and
in Section~\ref{sec:Calculus} we give the full set of inference rules of the
\kl{flower calculus} as well as our notion of proof. In
Section~\ref{sec:Semantics} we give a direct Kripke semantics to \kl{flowers},
and in Section~\ref{sec:Completeness} we show that a restricted fragment of
\kl{analytic} and invertible rules is complete with respect to the semantics.
Finally we conclude in Section~\ref{sec:Conclusion} by a comparison with some
related works.

\begin{note*}
  % The long version of this paper with complete appendices is available on arXiv
  % \cite{flower-calculus}.
  The proof of soundness of the \kl{flower calculus} is given in Appendix
  \ref{app:Soundness}. Contrary to the completeness proof, it is mostly routine
  work that does not require much insight. Detailed proofs for the deduction and
  completeness theorems are given respectively in Appendix
  \ref{app:Proofs-deduction} and Appendix \ref{app:Proofs-completeness}. Readers
  already familiar with \kl{EGs} can find a detailed comparison of the rules of
  the \kl{flower calculus} with Peirce's \kl{illative transformations} in
  Appendix \ref{app:Comparison}.
\end{note*}

\newpage

\section{Existential graphs}\label{sec:EGs}

\AP
Peirce designed in total three systems of \kl{EGs}, which he called respectively
\intro{Alpha}, \intro{Beta} and \intro{Gamma}. They were invented
chronologically in that order, which also captures their relationship in terms
of complexity: \kl{Alpha} is the foundation on which the other systems are
built, and can today be understood as a diagrammatic calculus for classical
\emph{propositional} logic. As we will see in Section~\ref{sec:Quantifiers},
\kl{Beta} corresponds to a variable-free representation of \emph{first-order}
logic without function symbols. The last system \kl{Gamma} is more
experimental, with various unfinished features that have been interpreted as
attempts to capture \emph{modal} \cite{zeman_graphical_1964} and
\emph{higher-order} logics.

\subparagraph*{Sheet of Assertion}

\AP
The most fundamental concept of \kl{Alpha} is the \intro{sheet of assertion},
denoted by $\intro*\SA$ thereafter. It is the space where statements are scribed
by the reasoner, typically a sheet of paper, a blackboard, or a computer
display. As its name indicates, scribing a statement on $\SA$ amounts to
\emph{asserting its truth}. Thus naturally, the empty $\SA$ where nothing is
scribed will denote \intro{vacuous truth}, traditionally signified by the symbol
$\intro*\True$.

\subparagraph*{Juxtaposition}

\AP
As we know from natural deduction, asserting the truth of the
\intro{conjunction} $a \intro*\Conj b$ of two propositions $a$ and $b$, amounts
to asserting \emph{both} the truth of $a$ and the truth of $b$. In \kl{Alpha},
there is no need to introduce the symbolic connective $\Conj$, since one can
just write both $a$ and $b$ at distinct locations on $\SA$:
\begin{mathpar}
  a~~~b
\end{mathpar}
More generally, one might consider any two portions $G$ and $H$ of $\SA$, and
interpret their \intro{juxtaposition} $G~H$ as signifying that we assert the
truth of their \kl{conjunction}.

\subparagraph*{Cuts}

\AP
Asserting the truth of the \intro{negation} $\intro*\Neg a$ of a proposition
$a$, amounts to \emph{denying} the truth of $a$. This is done in \kl{Alpha} by
\emph{enclosing} $a$ in a closed curve like so:
\begin{mathpar}
  \pcut{a}
\end{mathpar}
Peirce called such curves \intro{cuts}\footnote{Not to be confused with the name
given to instances of the \emph{cut rule} in sequent calculus.}, because they
ought to be seen as literal cuts in the paper sheet that embodies $\SA$. Note
that they do not need to be circles: all that matters is that $a$ is in a
separate area from the rest of $\SA$. This is precisely the content of the
\emph{Jordan curve theorem} in topology, and thus we can take \kl{cuts} to be
arbitrary Jordan curves. This entails in particular that \kl{cuts} cannot
intersect each other, but can be freely nested. Then as for \kl{juxtaposition},
one can replace the proposition $a$ in the interior of the \kl{cut} by any
\intro{graph} $G$ --- i.e. any portion of $\SA$ --- as long as the \kl{cut} does
not intersect other \kl{cuts} in $G$.

\subparagraph*{Relationship with formulas}

With just these two \emph{icons}, \kl{juxtaposition} and \kl{cuts}, one can
therefore assert the truth of any proposition made up of \kl{conjunctions} and
\kl{negations} and built from atomic propositions. Importantly, the only symbols
needed for doing so are letters $a, b, c\ldots$ denoting atomic propositions,
that is ``pure'' symbols that do not have any logical meaning associated to
them.

\AP
Now, it is well-known that $\{\Conj,\Neg\}$ is \emph{functionally complete},
meaning that any boolean truth function can be expressed as the composition of
\kl{conjunctions} and \kl{negations}. In particular, the symbolic definitions of
\intro{falsehood} $\intro*\False \defeq \Neg\True$, classical
\intro{disjunction} $A \intro*\Disj B \defeq \Neg(\Neg A \Conj \Neg B)$ and
classical \intro{implication} $A \intro*\Impl B \defeq \Neg (A \Conj \Neg B)$
can be expressed by the following three \kl{graphs}:
\begin{mathpar}
  \pcut{\phantom{A}}
  \and
  \pcut{\pcut{A}~~~\pcut{B}}
  \and
  \pcut{A~~~\pcut{B}}
\end{mathpar}
Thus one can easily encode any propositional formula into a classically
equivalent \kl{graph}. Conversely, one can translate any \kl{graph} into a
classically equivalent formula, as has been shown for instance in
\cite{10.7551/mitpress/3633.001.0001}. In fact, there are usually many possible
formula readings of a given \kl{graph}. One reason is that \kl{juxtaposition} of
\kl{graphs} is a \emph{variadic} operation, as opposed to \kl{conjunction} of
formulas which is \emph{dyadic}: thus formulas that only differ up to
\emph{associativity} are associated to the same \kl{graph}. Also, thanks to the
topological nature of $\SA$, \kl{juxtaposition} is naturally \emph{commutative}:
the locations of two juxtaposed \kl{graphs} do not matter, as long as they live in
the same area delimited by a \kl{cut}. The combination of these properties is
called the \intro{isotropy} of $\SA$ in \cite{minghui_graphical_2019}, and is
captured in traditional proof theory through the use of \emph{(multi)sets} for
modelling contexts in sequents.

% Hence, the graphs of \kl{Alpha} can be seen either as an
% associative-commutative normal form for propositional formulas built from atoms
% with $\{\Conj,\Neg\}$, or as a kind of \emph{nested sequent} for classical
% logic based on \kl{negation} rather than implication.

\subparagraph*{Illative transformations}

\AP
In order to have a proof system, one needs a collection of \emph{inference
rules} for deducing true statements from other true statements. In \kl{Alpha},
inference rules are implemented by what Peirce called \intro{illative
transformations} on \kl{graphs}. In modern terminology, they correspond to
\emph{rewriting} rules that can be applied to any \kl{subgraph}. By measuring
the depth of a \kl{subgraph} as the number of \kl{cuts} in which it is enclosed, we
thus have that the rules of \kl{Alpha} are applicable on \kl{subgraphs} of arbitrary
depth. This makes \kl{Alpha} deserving of the title of \emph{deep inference}
system.

\begin{figure}
  \begin{mathpar}
    \pcut{\pcut{a}~~~\pcut{\pcut{a}}}
    ~\xstep{\kl{Iteration}}~
    \pcut{\pcut{a}~~~\pcut{\phantom{\pcut{a}}}}
    ~\xstep{\kl{Insertion}}~
    \pcut{\phantom{\pcut{a}}~~~\pcut{\phantom{\pcut{a}}}}
    ~\xstep{\kl{Double{-}cut}}~
  \end{mathpar}
  \caption{Proof of the law of excluded middle in \kl{Alpha}}
  \label{fig:LEM}
\end{figure}

\AP
Figure \ref{fig:LEM} shows a proof of the law of excluded middle $a \Disj \Neg
a$ in \kl{Alpha}. The first step consists in applying the \kl{illative
transformation} of \intro{Iteration} to erase the \kl{subgraph} $\pcut{a}$. More
generally, \kl{Iteration} allows to erase any \kl{subgraph} $G$ as long as $G$
already occurs ``higher'' in $\SA$, i.e. in an area that encloses the erased
occurrence of $G$. The second step of \intro{Insertion} allows to erase the
other occurrence of $\pcut{a}$ because it is scribed in a \emph{negative} area,
i.e. an area enclosed in an \emph{odd} number of \kl{cuts} --- 1 in this
case\footnote{It might be quite confusing that we call ``Insertion'' a
transformation that \emph{erases} information. This is because we use Peirce's
original terminology, despite the fact that we adopt a \emph{backward} reading
of rules where the conclusion that we want to prove is reduced to a sufficient
premiss.}. The last step of \intro{Double{-}cut} allows to \emph{collapse} the
two remaining \kl{cuts}, because there is nothing but empty space in between
them. This leaves us with the empty $\SA$, having thus reduced the initial goal
to trivial truth.

\section{Flowers}\label{sec:Flowers}

\subparagraph*{The scroll}

\AP
In \cite[pp.~533--535]{peirce_prolegomena_1906}, Peirce explains that he did not
immediately come up with the idea of \kl{juxtaposition} and \kl{cuts} as
diagrammatizations of \kl{conjunction} and \kl{negation}. Instead, they arose as
the natural development of a more primitive icon that he called the
\intro{scroll}. Figure \ref{fig:scroll} shows Peirce's drawing of the
\kl{scroll} as it appears in \cite[Fig.~5]{peirce_prolegomena_1906}. He defines
its intended meaning as that of a ``conditional de inesse'', which corresponds
to the material \kl{implication} of classical logic. Then the \kl{graph} of Figure
\ref{fig:scroll} is interpreted as the formula $(A \Conj B) \Impl (C \Conj D)$.
This agrees with the encoding of \kl{implication} given in
Section~\ref{sec:EGs}, if one sees the outer boundary enclosing the antecedent
$A~B$ and the inner boundary enclosing the consequent $C~D$ as nested \kl{cuts}.

It is no coincidence that Peirce based his most fundamental icon on
\kl{implication}: according to Lewis \cite[p.~79]{Lewis1920-LEWASO-4}, he was the one
who introduced the ``illative relation'' of \kl{implication} into symbolic logic in
the first place, by giving it a distinguished symbol and studying extensively
the algebraic laws that govern it (e.g.\ Peirce's law $((A \Impl B) \Impl A)
\Impl A$).

\begin{figure}
  \captionsetup[subfigure]{justification=centering}
  \centering
  \begin{subfigure}[b]{0.22\textwidth}
    \centering
    \includegraphics[width=0.85\textwidth]{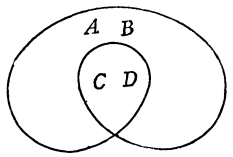}
    \caption{Peirce's \kl{scroll}}
    \label{fig:scroll}
  \end{subfigure}
  \begin{subfigure}[b]{0.24\textwidth}
    \centering
    \input{1.tikz}{0.4}{five-loops}
    \caption{Oostra's \kl{curl}}
    \label{fig:five-loops}
  \end{subfigure}
  \begin{subfigure}[b]{0.26\textwidth}
    \centering
    \input{1.tikz}{0.35}{scroll-inside-out}
    \caption{Inside-out \kl{curl}}
    \label{fig:scroll-inside-out}
  \end{subfigure}
  \begin{subfigure}[b]{0.22\textwidth}
    \centering
    \input{1.tikz}{1}{five-flower}
    \caption{\kl(informal){Flower}}
    \label{fig:five-flower}
  \end{subfigure}
  % \caption{Scrolls}
  \caption{From \kl{scrolls} to \kl(informal){flowers}}
\end{figure}

\subparagraph*{The $n$-ary scroll}

\AP In order to interpret the \kl{scroll} as an \emph{intuitionistic}
\kl{implication}, Oostra proposed in \cite{oostra_graficos_2010} to reify the
\kl{scroll} as a primitive icon of \kl{EGs}, distinguished from the nesting of
two \kl{cuts}. In fact he went further, by generalizing both the \kl{cut} and
the \kl{scroll} into an $n$-ary construction called the \intro{curl}, where $n$
is the number of inner boundaries, called \emph{loops}.
Figure~\ref{fig:five-loops} shows an example of \kl{curl} with five loops, where
the unique intersection points between inner and outer boundaries are
highlighted in \emph{orange}\footnote{We also shade the negative area delimited
by the outer boundary in \emph{gray}.}. In \cite{minghui_graphical_2019}, the
\kl{curl} is simply called \emph{$n$-ary \kl{scroll}}, the outer boundary
\intro{outloop}, and the inner boundaries \intro{inloops}. Then \kl{cuts} and
\kl{scrolls} are indeed special cases of $n$-ary \kl{scrolls}, respectively with
$n = 0$ and $n = 1$.

Like the unary \kl{scroll}, the $n$-ary \kl{scroll} is to be read as an
\kl{implication} whose antecedent is the content of the \kl{outloop}, and consequent
the content of the \kl{inloops}. The generalization consists in taking the
\emph{\kl{disjunction}} of the contents of all \kl{inloops}: this reflects nicely the
etymological meaning of the word ``disjunction'', since the \kl{inloops} enclose
\emph{disjoint} areas of the \kl{outloop} to which they are attached. Then the
$5$-ary \kl{scroll} of Figure \ref{fig:five-loops} can be read as the formula $a
\Impl (b \Disj c \Disj d \Disj e \Disj f)$; and the $0$-ary \kl{scroll} obtained
by removing all \kl{inloops} from the latter as $a \Impl \False$, since a $0$-ary
\kl{disjunction} is naturally evaluated to its neutral element $\False$. This
coincides with the intuitionistic reading of \kl{negation} $\Neg A \defeq A \Impl
\False$.
%, and is thus consistent with the interpretation of cuts as negations.

\subparagraph*{Continuity}

With this interpretation of the $n$-ary \kl{scroll}, the \kl{Alpha} encodings of
\kl{disjunction} and \kl{implication} as nested \kl{cuts} given in
Section~\ref{sec:EGs} are no longer valid, because they are not
intuitionistically equivalent to the associated binary and unary \kl{scrolls}.
This is illustrated in Figure \ref{fig:eg-disj-imp}, where the closeness in
meaning is reflected iconically (but not symbolically) in the fact that the
\kl{graphs} only differ in the \emph{continuity} (or lack thereof) between \kl{inloops}
and their \kl{outloop}.

\begin{remark}
This might be related to other manifestations of the notion of continuity in the
semantics of intuitionistic logic, such as the well-known Stone-Tarski
interpretation of formulas as topological spaces \cite{stone_topological_1938},
and the interpretation of proofs as continuous maps in the \emph{denotational
semantics} of Dana Scott\footnote{Before the advent of Oostra's intuitionistic
\kl{EGs}, Zalamea gave a detailed analysis of Peirce's philosophy of the
\emph{continuum}, how it relates to modern developments in mathematics, and how
it is embodied in \kl{EGs}~\cite{zalamea_peirces_2003}. Actually according to
Oostra \cite[p.~162]{oostra_advances_2022}, ``the possibility of developing
intuitionistic \kl{existential graphs} was first suggested by Zalamea in the
1990s \cite{zalamea_ieg_1,zalamea_ieg_2}''.} \cite{10.5555/218742.218744}.
% Thus the ($n$-ary) \kl{scroll} is a powerful icon, because it captures the
% distinction between classical and intuitionistic logic as being a matter of
% continuity between the space of inputs/hypotheses (\kl{outloop}) and the spaces of
% outputs/conclusions (\kl{inloops}), reflecting an intuition discovered much later in
% the denotational semantics of the $\lambda$-calculus.
\end{remark}

\begin{figure}
  \input{eg-disj-imp.tex}
  \caption{Continuity, \kl{disjunction} and \kl{implication} in intuitionistic \kl{EGs}}
  \label{fig:eg-disj-imp}
\end{figure}

\subparagraph*{Blooming}

\AP
In terms of ergonomy, the $n$-ary \kl{scroll} has one notable flaw, also shared
with the classical \kl{cut}-based syntax: it quickly induces heavy nestings of
curves in the plane, making even relatively simple \kl{graphs} hard to read for an
untrained eye.
% Before devising an alternative syntax, one should ask: what are the essential
% features of the \kl{scroll} that we want to preserve? Following the previous
% observations, we identified two of them:
% \begin{description}
%   \item[Continuity] the \kl{scroll} is a self-intersecting continuous curve,
%   which can be drawn in one stroke of the pen;
%   \item[Polarity] this curve delineates two kinds of areas: \kl{inloops}
%   that have the same polarity as the area on which the \kl{scroll} is scribed, and
%   the \kl{outloop} which has the opposite polarity.
% \end{description}
% Fortunately, these two properties are preserved when turning \kl{inloops}
Our solution is to turn \kl{inloops} \emph{inside-out}, as illustrated in Figure
\ref{fig:scroll-inside-out}. In this way, we effectively divide the amount of
curve-nesting in \kl{scrolls} by two. And as an added bonus, the new icon is
reminiscent of a \intro(informal){flower}, as if it had bloomed from its
curled bud; or as if the pistol cylinder from Figure \ref{fig:five-loops} had
transformed into a \emph{pistil}, and its bullet chambers into \emph{petals}.
% \footnote{As the saying goes: make love, not war.}.

From that point onwards, we chose to fully embrace the flower metaphor: first in
our drawing style as witnessed in Figure \ref{fig:five-flower}, but also in our
syntactic terminology, to be introduced in the next pages. Negative (resp.
positive) \kl{outloops} are now drawn as \emph{yellow} (resp. \emph{white})
pistils for a slightly more colorful experience, and \kl{inloops} as transparent
petals, i.e. of the same color as the area on which they are scribed.
% We also drop the requirement that petals should intersect their pistil at a
% single point, for purely aesthetic reasons.

\section{Gardens}\label{sec:Quantifiers}

\subparagraph*{Lines of identity}

\AP To express first-order quantification, Peirce introduced in \kl{Beta} the
icon of \intro{lines of identity} (\reintro{LoIs}). In short, the usual binders
and variables of predicate calculus are replaced by \emph{lines} that connect
the occurrences of bound variables in predicate arguments to their binding
point. For instance, the formulas $\exists x. P(x) \Conj Q(x)$ and $\forall x.
R(x) \Impl S(x)$ can be represented in \kl{Beta} by the \kl{graphs} of Figure
\ref{fig:loi}.

The kind of quantification is determined by the location of the binding point,
which is taken to be the \emph{outermost} point in the line: if it is in a
\emph{positive} area as in the upper \kl{graph}, then the quantifier is
\emph{existential}; otherwise if it is in a \emph{negative} area as in the lower
\kl{graph}, the quantifier is \emph{universal}. This is justified by De Morgan's
laws: the lower \kl{graph} can also be read as the classically equivalent formula
$\Neg \exists x. R(x) \Conj \Neg S(x)$.

\subparagraph*{Intuitionistic quantification}

In intuitionistic logic however, De Morgan's laws do not hold anymore. Thus in
the \kl{flower calculus} we need a different way to interpret \kl{LoIs} as
quantifiers. Our key insight is to adopt a \emph{polarity-invariant} viewpoint:
a \kl{LoI} now has \emph{existential} (resp. \emph{universal}) force when its
outermost point is located in a \emph{petal} (resp. \emph{pistil}). In
particular, this implies that \kl{LoIs} cannot occur at the top-level of $\SA$
anymore, but only inside \kl(informal){flowers}. Thus the two previous \kl{Beta}
\kl{graphs} are transformed into the single-petal \kl(informal){flowers} of Figure
\ref{fig:loi-flower}.

\subparagraph*{Variables}

Quine experimented with a notation similar to \kl{LoIs}, but deemed it ``too
cumbersome for practical use'' \cite[p.~125]{Roberts+1973}. While his lines
connected locations inside symbolic formulas written in linear notation, it is
true that having a line for each occurrence of bound variable can quickly lead
to unreadable diagrams ridden with overlapping lines.
%\footnote{This problem is
%well-known in the realm of node graphs-based visual programming languages.}.
This is not a problem in the context of Peirce's work, because his aim was ``to
separate [relational] reasoning into its smallest steps, [...] not to facilitate
reasoning, but to facilitate the study of reasoning''
\cite[p.~111]{Roberts+1973}; and recent formalizations of the algebra of \kl{LoIs} in
category theory support the pertinence of Peirce's approach
\cite{pietarinen_compositional_2020,bonchi_diagrammatic_2024}.

\begin{figure}
  \captionsetup[subfigure]{justification=centering}
  \centering
  \begin{subfigure}{0.32\textwidth}
    \centering
    \input{1.tikz}{0.4}{loi-ex}
    \\[1em]
    \input{1.tikz}{0.4}{loi-fa}
    \caption{\kl{EGs} with \kl{LoIs}}
    \label{fig:loi}
  \end{subfigure}
  \begin{subfigure}{0.33\textwidth}
    \centering
    \input{1.tikz}{1}{loi-ex-flower}
    \\[1em]
    \input{1.tikz}{1}{loi-fa-flower}
    \caption{Flowers with \kl{LoIs}}
    \label{fig:loi-flower}
  \end{subfigure}
  \begin{subfigure}{0.33\textwidth}
    \centering
    \input{1.tikz}{1}{sprinkler-ex}
    \\[1em]
    \input{1.tikz}{1}{sprinkler-fa}
    \caption{Flowers with variables}
    \label{fig:sprinkler-flower}
  \end{subfigure}
  \caption{From \kl{LoIs} to variables}
\end{figure}

However, keeping in mind our goal of laying the basis for a calculus well-suited
to practical reasoning in \kl{ITPs}, we chose to replace \kl{LoIs} by a more
traditional syntax based on binders and variables. The idea is to substitute
every \kl{LoI} with a variable \reintro{binder} scribed in the area of its
outermost point, so that the two \kl(informal){flowers} of Figure
\ref{fig:loi-flower} transform into those of Figure \ref{fig:sprinkler-flower}.
Areas delimited by pistils and petals now comprise both \kl(informal){flowers}
and binders, which can be seen metaphorically as \reintro{sprinklers} that
irrigate the leaves (atomic predicates) of flowers through invisible \kl{LoIs},
imagined as underground hoses. Hence we call these areas \reintro{gardens}.
% Every area inside a pistil or petal now contains both a multiset of (nested)
% flowers, and a set of bound variables. Building on our botanical metaphor, we
% call finite flower multisets \emph{bouquets} and finite variable sets
% \emph{sprinklers}: indeed, variables can be seen as irrigating the leaves
% (atomic predicates) of flowers through the (now invisible) \kl{LoIs}, viewed as
% hoses. The combined data of a bouquet $\Phi$ and a sprinkler $\bx$ that makes up
% the content of an area is called a \emph{garden}, and written
% $\garden{\bx}{\Phi}$. Then we can write any flower as an expression
% $\flower{\gamma}{\Delta}$, where $\gamma$ is the garden in its pistil and
% $\Delta$ the multiset of gardens in its petals, also called its \emph{corolla}.

\newpage

\section{Syntax}\label{sec:Syntax}

We are now going to distill the syntactic essence of \kl(informal){flowers} into
an inductive, (multi)set-based data structure. This will allow for a more
compact textual notation, that is better suited to proof-theoretical study. We
previously illustrated how \kl(informal){flowers} allow to represent purely
relational statements without function symbols. Since functions are just
deterministic relations, one can in principle formalize any first-order theory
in this syntax\footnote{Conversely, every relation can be faithfully encoded as
its characteristic function, which is the basis for the formalization of
mathematics in \emph{type theories}.}.

\begin{definition}%[First-order signature] % pas la peine vu qu'il y a le \emph déjà
\AP
  A \intro{first-order signature} is a pair $\Sigma = ( \intro*\psymbs,
  \intro*\arity )$, where $\psymbs$ is the countable set of \intro{predicate
  symbols} of $\Sigma$, and $\arity : \psymbs \to \nats$ gives an \intro{arity}
  to each symbol.
\end{definition}

\AP
In the following, we fix a countable set of variables $\intro*\vars$
and a first-order signature $\Sigma$.

\begin{definition}%[Flowers]
  \label{def:flowers}
  The sets of \intro{flowers} $\intro*\flowers$ and \intro{gardens}
  $\intro*\gardens$ are defined by mutual induction:
  \begin{description}
    \item[Atom] If $p \in \psymbs$ and $\tvec{x} \in \vars^{\arity(p)}$, then
    $p(\tvec{x}) \in \flowers$;
    \itemAP[Garden] If $\bx \subset \vars$ is a finite set and $\Phi \subset
    \flowers$ a finite multiset, then $\intro*\garden{\bx}{\Phi} \in \gardens$;
    \itemAP[Flower] If $\gamma \in \gardens$ and $\Delta \subset \gardens$ is a
    finite multiset, then $\intro*\flower{\gamma}{\Delta} \in \flowers$.
  \end{description}
  Similarly to \intro{nested sequents}, the syntax of \kl{flowers} $\phi, \psi$
  and \kl{gardens} $\gamma, \delta$ can be expressed succinctly with the
  following grammar:
  \begin{mathpar}
    \phi, \psi \Coloneq p(x_1, \ldots, x_n) \mid \flower{\gamma}{\delta_1 \sep \ldots \sep \delta_n} \and
    \gamma, \delta \Coloneq \garden{x_1, \ldots, x_n}{\phi_1, \ldots, \phi_n}
  \end{mathpar}
\end{definition}

\begin{figure}
  \captionsetup[subfigure]{justification=centering}
  \centering
  \begin{subfigure}{0.35\textwidth}
    \centering
    \begin{tabular}{|c|c|}
      \hline
      \bfseries Kind & \bfseries Letters \\
      \hline
      Variables ($\vars$) & $x, y, z$ \\
      \kl{Flowers} ($\flowers$) & $\phi, \psi, \xi$ \\
      \kl{Gardens} ($\gardens$) & $\gamma, \delta$ \\
      \kl{Sprinklers} & $\bx, \by, \bz$ \\
      Variable vectors & $\tvec{x}, \tvec{y}, \tvec{z}$ \\
      \kl{Substitutions} & $\sigma, \tau$ \\
      \kl{Bouquets} & $\Phi, \Psi, \Xi$ \\
      \kl{Corollas} & $\Gamma, \Delta$ \\
      \kl{Contexts} & $\cPhi, \cPsi, \cXi$ \\
      \kl{Theories} & $\mathcall{T}, \mathcall{U}$ \\
      \hline
    \end{tabular}
    \caption{Conventions for meta-variables}
    \label{tab:letters}
  \end{subfigure}
  \begin{subfigure}{0.64\textwidth}
    \centering
    \stkfig{0.85}{parsed-flower}
    \caption{Interpreting flowers}
    \label{fig:parsed-flower}
  \end{subfigure}
  \caption{Notations}
\end{figure}

\AP Building on our botanical metaphor, any finite set $\bx \subset \vars$ of
variables is called a \intro{sprinkler}, finite multiset $\Phi \subset \flowers$
of \kl{flowers} a \intro{bouquet}, and finite multiset $\Gamma \subset \gardens$
of \kl{gardens} a \intro{corolla}. Following the grammar presentation, we will
often write \kl{gardens} as $\garden{x_1, \ldots, x_n}{\phi_1, \ldots, \phi_m}$,
where the $x_i$ are called \intro{binders}; and non-atomic \kl{flowers} as
$\flower{\gamma}{\delta_1 \sep \ldots \sep \delta_n}$, where $\gamma$ is the
\intro{pistil} and the $\delta_i$ are the \intro{petals}. We write
$\intro*\fset{i}{n}{E_i}$ to denote a finite (multi)set of size $n$ with
elements $E_i$ indexed by $1 \leq i \leq n$. We also omit writing the empty
(multi)set, accounting for it with blank space as is done in sequent notation;
in particular, $\garden{{}}{{}}$ stands for the empty \kl{garden}
$\garden{\emptyset}{\emptyset}$, $\flower{\gamma}{{}}$ for the \kl{flower} with
no \kl{petals} $\flower{\gamma}{\emptyset}$, and
$\flower{\gamma}{\garden{{}}{{}}}$ for the \kl{flower} with one empty
\kl{petal}.

Note that the order of precedence of operators is $\mathop{,} < \garden{{}}{{}}
< \mathop{;} < \mathop{\flower{}{}}$: this is illustrated in Figure
\ref{fig:parsed-flower}, where a \kl[flower]{flower expression} is parsed into
the corresponding \kl[flower](informal){flower drawing}, and then translated as
a formula. Also to improve readability, we will most of the time omit the \kl{garden}
dot `$\cdot$' when the \kl{sprinkler} is empty, writing $\Phi$ instead of
$\garden{{}}{\Phi}$.

\begin{remark}
  In some places the choice of letter for meta-variables will be important to
  disambiguate the kind of syntactic object we denote. Table \ref{tab:letters}
  summarizes our chosen notational conventions in this respect.
\end{remark}

\AP
We now proceed with routine definitions for handling variables.

\begin{definition}%[Free and bound variables]
  \label{def:fv}
  The sets of \intro{free variables} $\intro*\fv(-)$ and \intro{bound variables}
  $\intro*\bv(-)$ of a \kl{flower}/\kl{bouquet}/\kl{garden} are defined
  recursively by:
  \begin{mathpar}
    \fv(p(\tvec{x})) = \tvec{x} \and
    \fv(\Phi) = \bigcup_{\phi \in \Phi}{\fv(\phi)} \and
    \fv(\garden{\bx}{\Phi}) = \fv(\Phi) \setminus \bx \and
    \fv(\flower{\garden{\bx}{\Phi}}{\Delta}) = \fv(\garden{\bx}{\Phi}) \cup \bigcup_{\garden{\by}{\Psi} \in \Delta}{\fv(\garden{\bx, \by}{\Psi})}
  \end{mathpar}
  % We say that a term/flower/bouquet/garden is \emph{closed} when its set of free
  % variables is empty.
  % The sets of closed terms, flowers and gardens are denoted respectively by
  % $\closed{\terms}$, $\closed{\flowers}$ and $\closed{\gardens}$.
%\end{definition}
%
% \begin{remark}\label{rem:flower-scope}
% Note that the scope of a binder located in a pistil extends both to the pistil
% \emph{and} to all its attached petals, whereas for a binder located in a petal
% it is limited to said petal. This is reflected in the above definition of free
% variables for (non-atomic) flowers, and is visually explained by the nesting of
% curves in an $n$-ary \kl{$1}.
% \end{remark}
%
%\begin{definition}
  % The sets of bound variables $\bv(-)$ of a flower/bouquet/garden are defined:
  \begin{mathpar}
    \bv(p(\tvec{x})) = \emptyset \quad\,\,
    \bv(\Phi) = \bigcup_{\phi \in \Phi}{\bv(\phi)} \quad\,\,
    \bv(\garden{\bx}{\Phi}) = \bx \cup \bv(\Phi) \quad\,\,
    \bv(\flower{\gamma}{\Delta}) = \bv(\gamma) \cup \bigcup_{\delta \in \Delta}{\bv(\delta)}
  \end{mathpar}
\end{definition}

\AP
To avoid reasoning about $\alpha$-equivalence, we adopt in this work the
so-called \intro{Barendregt convention} that all variable \kl{binders} are distinct,
both among themselves and from free variables. Formally, we assume that for any
\kl{bouquet} $\Phi$ the two following conditions hold:
\begin{enumerate}
  \item computing $\bv(\Phi)$ as a multiset gives the same result as computing
  it as a set;
  \item $\bv(\Phi) \cap \fv(\Phi) = \emptyset$.
\end{enumerate}

\AP
To define substitutions, we introduce a general notion of \emph{function
update}, which will be useful for the semantic evaluation of \kl{flowers} in
Section \ref{sec:Semantics}.

\begin{definition}%[Function update]
  \label{def:update}
  Let $A, B$ be two sets, $f, g : A \to B$ two functions and $R \subseteq A$
  some subset of their domain. The \intro{update} of $f$ on $R$ with $g$ is the
  function defined by:
  \begin{align*}
  (f \intro*\upd{R} g)(x) =
  \begin{cases}
    g(x) &\text{if $x \in R$} \\
    f(x) &\text{otherwise}
  \end{cases}
  \end{align*}
  $\update{-}{-}{-}$ is left-associative, that is
  $\update{f}{R}{\update{g}{S}{h}} = \update{(\update{f}{R}{g})}{S}{h}$. Also
  if $f$ or $g$ is the identity function $\idsubst$ we omit writing it, i.e.
  $\update{f}{R}{{}} = \update{f}{R}{\idsubst}$ and $\update{{}}{R}{g} =
  \update{\idsubst}{R}{g}$.
\end{definition}

\begin{definition}%[Substitution]
\AP
  A \intro{substitution} is a function $\sigma : \vars \to \vars$ with a finite
  \intro{support} $\intro*\dom(\sigma) = \compr{x}{\sigma(x) \not= x}$. We write
  $\sigma : \bx $ to denote a \kl{substitution} $\sigma$ whose support is $\bx$.
  The domain of \kl{substitutions} is extended to \kl{flowers}, \kl{bouquets}
  and \kl{gardens} mutually recursively by:
  \begin{mathpar}
    \sigma(p(x_1, \ldots, x_n)) = p(\sigma(x_1), \ldots, \sigma(x_n))
    \and
    \sigma(\phi_1, \ldots, \phi_n) = \sigma(\phi_1), \ldots, \sigma(\phi_n)
    \and
    \sigma(\garden{\mathbf{x}}{\Phi}) =
      \garden{\mathbf{x}}{\restr{\sigma}{\mathbf{x}}(\Phi)}
    \and
    \sigma(\flower{\garden{\mathbf{x}}{\Phi}}{\delta_1 \sep \ldots \sep
      \delta_n}) =
      \flower{\sigma(\garden{\mathbf{x}}{\Phi})}{\restr{\sigma}{\mathbf{x}}(\delta_1)
      \sep \ldots \sep \restr{\sigma}{\mathbf{x}}(\delta_n)}
  \end{mathpar}
% \end{definition}
% \begin{definition}%[Capture-avoiding substitution]
\AP
  We say that a \kl{substitution} $\sigma : \bx$ is \intro{capture-avoiding} in a
  \kl{bouquet} $\Phi$ if $\sigma(\bx) \cap \bv(\Phi) = \emptyset$.
\end{definition}

\section{Calculus}\label{sec:Calculus}

%\subparagraph*{Contexts}

\AP
Equipped with an inductive syntax, we can now express formally the inference
rules of the \kl{flower calculus}. First we need a notion of \emph{context} to
apply rules at arbitrarily deep locations:

\begin{definition}[Context]\label{def:contexts}
  \intro{Contexts} $\cPhi$ are defined inductively by the following grammar:
  \begin{mathpar}
    \cPhi, \cPsi, \cXi \Coloneq \Psi, \cphi
    \and
    \cphi, \cpsi, \cxi \Coloneq \hole \mid \flower{\garden{\bx}{\cPhi}}{\Delta} \mid \flower{\gamma}{\garden{\bx}{\cPhi} \sep \Delta}
  \end{mathpar}
  Informally, a \kl{context} can be seen as a \kl{bouquet} with exactly one
  occurrence of a special \kl{flower} $\intro*\hole$ called its \emph{hole}. The
  \intro{filling} of a \kl{context} $\cPhi$ with a \kl{bouquet} $\Psi$ (resp.
  \kl{context} $\cPsi$) is the \kl{bouquet} $\intro*\cfill{\cPhi}{\Psi}$ (resp.
  \kl{context} $\cfill{\cPhi}{\cPsi}$) equal to $\cPhi$ where $\hole$ has been
  substituted with $\Psi$ (resp. $\cPsi$).
  % \begin{align*}
  %   \cfill{\hole}{\Phi} &= \Phi
  %   &
  %   \cfill{(\Psi, \cphi)}{\Phi} &= \Psi, \cfill{\cphi}{\Phi}
  %   \\
  %   \cfill{(\flower{\garden{\bx}{\cPhi}}{\Delta})}{\Phi} &=
  %   \flower{\garden{\bx}{\cfill{\cPhi}{\Phi}}}{\Delta}
  %   &
  %   \cfill{(\flower{\gamma}{\garden{\bx}{\cPhi} \sep \Delta})}{\Phi} &=
  %   \flower{\gamma}{\garden{\bx}{\cfill{\cPhi}{\Phi}} \sep \Delta}
  % \end{align*}
  % A context $\cPhi$ contains exactly one occurrence of the atomic context
  % $\hole$ called its \emph{hole}. The hole can be \emph{filled} (substituted)
  % with any bouquet $\Psi$ or context $\cXi$, producing a new bouquet
  % $\cfill{\Phi}{\Psi}$ or context $\cfill{\Phi}{\cXi}$. In particular,
  % filling with the empty bouquet yields a bouquet
  % $\cfill{\Phi}{}$, which is just $\cPhi$ with its hole
  % removed. 
\end{definition}

\begin{definition}[Polarity]
\AP
  The number of \emph{inversions} $\intro*\inv(\cPhi)$ of a \kl{context} $\cPhi$ is:
  \begin{mathpar}
    \inv(\hole) = 0
    \quad\,\,
    \inv(\Psi, \cphi) = \inv(\cphi)
    \quad\,\,
    \inv(\flower{\garden{\bx}{\cPhi}}{\Delta}) = 1 + \inv(\cPhi)
    \quad\,\,
    \inv(\flower{\gamma}{\garden{\bx}{\cPhi} \sep \Delta}) = \inv(\cPhi)
  \end{mathpar}
  We say that a \kl{context} $\cPhi$ is \intro{positive} if $\inv(\cPhi)$ is
  even, and \intro{negative} otherwise. We denote \kl{positive} and \kl{negative}
  \kl{contexts} respectively by $\cPhiP$ and $\cPhiN$.
\end{definition}

%\subparagraph*{Pollination}

\AP
In order to formulate the equivalent of the \kl{Iteration} rule of \kl{EGs} for
\kl{flowers}, we introduce a \emph{pollination} relation that captures the
availability of a \kl{flower} in a given \kl{context}:

\begin{definition}[Pollination]\label{def:pollination}
  We say that a \kl{flower} $\phi$ can be \intro{pollinated} in a \kl{context}
  $\cPhi$, written $\intro*\chyp{\phi}{\cPhi}$, when there exists a \kl{bouquet}
  $\Psi$ with $\phi \in \Psi$ and \kl{contexts} $\cXi$ and $\cXi_0$ s.t. either:
  \begin{description}
    \itemAP[\intro{Cross-pollination}] $\cPhi = \cfill{\cXi}{\Psi,
    \cXi_0}$;
    \itemAP[\intro{Self-pollination}] $\cPhi =
    \cfill{\cXi}{\flower{\garden{\bx}{\Psi}}{\garden{\by}{\cXi_0}
    \sep \Delta}}$ for some $\bx, \by, \Delta$.
  \end{description}
  A \kl{bouquet} $\Phi$ can be \reintro{pollinated} in $\cPhi$, written
  $\reintro*\chyp{\Phi}{\cPhi}$, if $\chyp{\phi}{\cPhi}$ for all $\phi \in
  \Phi$.
\end{definition}

\begin{figure}
  \captionsetup[subfigure]{justification=centering}
  \centering
  \begin{subfigure}{0.3\textwidth}
    \centering
    \input{0.9.tikz}{0.5}{cross-pollination}
    \caption{\kl{Cross-pollination}}
  \end{subfigure}
  \begin{subfigure}{0.3\textwidth}
    \centering
    \input{0.9.tikz}{0.5}{self-pollination}
    \caption{\kl{Self-pollination}}
  \end{subfigure}
  \caption{\kl{Pollination} in \kl{flowers}}
  \label{fig:pollination}
\end{figure}

Figure \ref{fig:pollination} illustrates the meaning of \kl{pollination} as a
relation of \emph{justification} between locations: the blue dot marks the
location of the justifying/\kl{pollinating} occurrence of $\phi$, and the red
dots all the areas that it justifies/\kl{pollinates}, and thus where $\phi$ is
available for use. We distinguish two cases of \kl{cross-pollination} and
\kl{self-pollination}, as botanists do when describing the reproduction of
\kl{flowers}. This distinction does not exist in classical \kl{EGs}, because
\kl{pistils} and \kl{petals} are both identified as instances of
\kl{cuts}\footnote{The same phenomenon is at work in \kl{SFL}:
\kl{cross-pollination} and \kl{self-pollination} can be seen as generalizing the
\emph{forward} and \emph{backward} interaction connectives $\forw$ and $\back$
of intuitionistic \kl{SFL}
\cite{DBLP:conf/cade/Chaudhuri21,10.1145/3497775.3503692}, while the original
formulation of \kl{SFL} for classical linear logic had only one interaction
connective $\ast$ \cite{Chaudhuri2013}. Through the Curry-Howard-Lambek
correspondence, this is also reminiscent of the adjunction between products
($\forw$) and exponentials ($\back$) in \emph{cartesian closed categories}, as
opposed to the natural isomorphism $(-)^*$ of \emph{$*$-autonomous categories}. }.

\begin{remark}
  Incidentally, the \kl{pollination} relation also explains the \emph{scope} of
  variables. Indeed, one can interpet red dots in Figure \ref{fig:pollination}
  as the allowed \emph{usage} points for the variable \emph{bound} at the linked
  blue dot. This hints at a possible \emph{type-theoretic} variant of the
  \kl{flower calculus} where variables are also used for \kl{higher-order}
  individuals, including \kl{flowers}.
\end{remark}

\subparagraph*{Proofs}

\begin{figure}
  \input{flower-calculus.tex}
  \caption{Rules of the \kl{flower calculus}}
  \label{fig:flower-calculus}
\end{figure}

% \begin{figure}
%   \input{natural-rules.tex}
%   \caption{Natural rules of the \kl{flower calculus} $\Nature$}
%   \label{fig:natural-rules}
% \end{figure}

The inference rules of the \kl{flower calculus} are presented in Figure
\ref{fig:flower-calculus}. Read from top to bottom, they correspond to
traditional inference rules deducing a necessary conclusion from a valid
premiss. But we will prefer their backward, \emph{bottom-up} reading: then they
can be seen as \emph{rewriting} rules that reduce a goal to a sufficient
premiss, just like in our illustration of the \kl{illative transformations} of
\kl{EGs} in Figure \ref{fig:LEM}. Also, all rules manipulate \emph{bouquets}:
this is seen more clearly in the \emph{graphical} presentation of the rules in
appendix (Figures \ref{fig:natural-graphical} and \ref{fig:cultural-graphical}).

\AP
We partition the rules into two sets: the \intro{natural} rules denoted by
$\Nature$ that apply in arbitrary \kl{contexts}, and the \intro{cultural} rules
denoted by $\Culture$ that apply exclusively in \kl{positive} or \kl{negative}
\kl{contexts}. In particular, every $\Nature$-rule is both \intro{analytic}
(i.e. every atom in the premiss already appears in the conclusion) and
\emph{invertible} (Lemma~\ref{lem:flowers-local-soundness}); on the contrary,
all $\Culture$-rules are \emph{non-invertible}, and they will be shown to be
\emph{admissible} in Section \ref{sec:Completeness}.

\begin{definition}[Derivation]\label{def:derivation}
\AP
  Given a set of rules $\mathsf{R}$, we write $\Phi \intro*\step{\mathsf{R}}
  \Psi$ to indicate a rewrite \emph{step} in $\mathsf{R}$, that is an instance
  of some $\mathsf{r} \in \mathsf{R}$ with $\Psi$ as premiss and $\Phi$ as
  conclusion. We just write $\Phi \reintro*\step{} \Psi$ to mean $\Phi
  \step{\Nature\cup\Culture} \Psi$. A \intro{derivation} $\Phi
  \intro*\nsteps{n}{\mathsf{R}} \Psi$ is a sequence of rewrite steps $\Phi_0
  \step{\mathsf{R}} \Phi_1 \ldots \step{\mathsf{R}} \Phi_n$ with $\Phi_0 =
  \Phi$, $\Phi_n = \Psi$ and $n \geq 0$. Generally the length $n$ of the
  \kl{derivation} does not matter, and we just write $\Phi
  \intro*\steps{\mathsf{R}} \Psi$. Finally, \kl{natural} \kl{derivations} are
  closed under arbitrary \kl{contexts}: for every \kl{context} $\cXi$, $\Phi
  \step{\Nature} \Psi$ implies $\cfill{\cXi}{\Phi} \step{\Nature}
  \cfill{\cXi}{\Psi}$. We write $\Phi \intro*\lstep \Psi$ to denote a
  \intro{shallow} \kl{natural} step, i.e. an instance of a $\Nature$-rule in the
  empty \kl{context} $\hole$.
\end{definition}

\begin{definition}[Proof]\label{def:flowers-proof}
\AP
  A \intro{proof} of a \kl{bouquet} $\Phi$ is a \kl{derivation} $\Phi \steps{} \emptyset$.
\end{definition}

In Peircean terms, the empty \kl{bouquet} is the blank $\SA$. Then proving a
\kl{bouquet} amounts to erasing it completely from $\SA$, thus reducing it to
\kl[vacuous truth]{trivial truth} as in Figure \ref{fig:LEM}. Figure
\ref{fig:flowers-proof-example} in appendix shows an example of
$\Nature$-\kl{proof} in the \kl{flower calculus}, both in textual and graphical
syntax. Note that we used a non-duplicating version of the rules \kl{ipis} and
\kl{ipet}, in order to save some space in the graphical presentation.

% \begin{figure}
%   \input{cultural-rules.tex}
%   \caption{Cultural rules of the \kl{flower calculus} $\Culture$}
%   \label{fig:cultural-rules}
% \end{figure}

If we want to reason about \emph{relative} truth, i.e. $\Phi$ is true under the
assumption that $\Psi$ is, we can simply rely on the existence of a \kl{derivation}
$\Phi \steps{} \Psi$ in the full \kl{flower calculus}. This will be justified by
the soundness of all rules (Theorem~\ref{thm:flowers-soundness}) as well as a
\emph{strong} completeness result
(Corollary~\ref{thm:flowers-strong-completeness}), that relies on the following
strong deduction theorem:

\begin{theorem}[Strong deduction]\label{thm:flowers-strong-deduction}
  $\Phi \steps{} \Psi$ if and only if $\flower{\Psi}{\Phi} \steps{} \emptyset$.
\end{theorem}

Contrary to full derivability, \kl{natural} derivability $\Phi \steps{\Nature}
\Psi$ is too weak to satisfy a strong deduction theorem. This is a consequence
of the fact that $\Nature$-rules are \emph{invertible}, and thus can only relate
\kl{equivalent} \kl{bouquets}. Indeed, as soon as $\flower{\Psi}{\Phi}$ is
$\Nature$-provable but the converse $\flower{\Phi}{\Psi}$ is not, it follows
from the completeness of $\Nature$-rules that $\Phi$ and $\Psi$ are not
\kl{equivalent}: thus $\Phi \notsteps{\Nature} \Psi$, contradicting the strong
deduction statement.

\AP A trivial way to circumvent this is to define directly the relation of
\emph{hypothetical provability} $\Psi \entails{} \Phi$ as $\flower{\Psi}{\Phi}
\steps{} \emptyset$. This is closer to what one would find in sequent calculus,
where hypothetical proofs are closed derivations of hypothetical sequents, not
open derivations. The difference is that sequents capture only the
\intro{first-order}\footnote{As opposed to \intro{higher-order}, in the sense of
having negatively nested \kl{implications}.} implicative structure of logic,
while \kl{flowers} capture the full structure of intuitionistic \kl{FOL}. This
allows for a nice generalization of the notion of hypothetical provability,
which will be useful in our completeness proof:

\begin{definition}
\AP
  We say that $\Phi$ is \intro{hypothetically provable} from $\Psi$ in a
  fragment $\mathsf{R}$ of rules, written $\Psi \intro*\entails{\mathsf{R}}
  \Phi$, if $\cfill{\cXi}{\Phi} \steps{\mathsf{R}} \cfill{\cXi}{}$ for every
  \kl{context} $\cXi$ such that $\chyp{\Psi}{\cXi}$. We write $\Psi
  \reintro*\entails{} \Phi$ to denote \kl{hypothetical provability} in the full
  \kl{flower calculus}.
\end{definition}

% There is a subtle but important shift here with respect to the standard notions
% of hypothetical provability, as found in Gentzen systems or type theories: while
% in these settings it is characterized as the existence of a proof for a
% \emph{single} hypothetical judgment $\Gamma \seq C$ which constrains the space
% of derivations, here we have the stronger requirement that there exist proofs
% for a \emph{class} of judgments $\cfill{\Xi}{\Phi}$, whose hypothetical shape
% comes from the condition that $\chyp{\Psi}{\cXi}$.

\begin{theorem}[Deduction]\label{thm:flowers-deduction}
  $\Psi \entails{\Nature} \Phi$ if and only if
  $\entails{\Nature} \flower{\Psi}{\Phi}$.
\end{theorem}

\section{Semantics}\label{sec:Semantics}

We now give a semantics to \kl{flowers} in Kripke structures. We recall the
standard definitions:

\begin{definition}
\AP
  A \intro{first-order structure} is a pair $( M, \intro*\interp{{\cdot}} )$ where
  $M$ is a non-empty set called the \intro{domain}, and $\interp{{\cdot}}$ is a
  map called the \intro{interpretation} that associates to each predicate symbol
  $p \in \psymbs$ a relation $\interp{p} \subseteq M^{\arity(p)}$.
\end{definition}

\begin{definition}
\AP
  A \intro{Kripke structure} is a triplet $\mathcal{K} = ( W, \intro*\access,
  (M_w)_{w \in W} )$, where $W$ is the set of \intro{worlds}, $\access$ is a
  pre-order on $W$ called \intro{accessibility}, and $(M_w)_{w \in W}$ is a
  family of \kl{first-order structures} indexed by $W$. Furthermore, we require
  the following monotonicity conditions to hold whenever $w \access w'$:
  \begin{enumerate*}
    \item $M_w \subseteq M_{w'}$;
    \item for every $p \in \psymbs$, $\interp{p}_w \subseteq
      \interp{p}_{w'}$.
  \end{enumerate*}
\end{definition}

% Then we need a way to interpret free variables in any given world $w$ of a
% Kripke structure, which is done through the concept of \emph{\evaluation{w}}:

\begin{definition}
\AP
  Given a \kl{Kripke structure} $\mathcal{K}$ and a \kl{world} $w$ in $\mathcal{K}$,
  a \intro*\evaluation{w} is a function $e : \vars \to M_w$.
  % We write $e : \bx \to M_w$ if $e(x) = \ast$ for every $x \not\in
  % \bx$ and some distinguished element $\ast \in M_w$.
  The \kl{interpretation} map of $M_w$ is extended to variables and
  \kl{substitutions} with respect to any \evaluation{w} $e$ as follows:
  \begin{mathpar}
  \interp{x}_e = e(x) \and
  \interp{\sigma}_e(x) = \interp{\sigma(x)}_e
  % \interp{p(t_1, \ldots, t_n)}_e &~\text{iff}~ (\interp{t_1}_e, \ldots, \interp{t_n}_e) \in \interp{p}_w \\
  \end{mathpar}
\end{definition}

The crux of Kripke semantics is the \emph{forcing} relation, that captures the
truth-conditions of statements in \kl{Kripke structures}. While it is usually
defined on formulas, here we adapt the definition to \kl{flowers}, which in our
opinion makes it simpler and more uniform since \kl{flowers} can be seen as
built from essentially one big constructor:

\begin{definition}
\AP
  The \intro{depth} $\intro*\fdepth{{-}}$ of a \kl{flower}/\kl{garden} is defined
  by mutual recursion:
  \begin{mathpar}
    \fdepth{p(\tvec{x})} = 0
    \and
    \fdepth{\garden{\bx}{\Phi}} = \max_{\phi \in \Phi}{\fdepth{\phi}}
    \and
    \fdepth{\flower{\gamma}{\Delta}} = 1 + \max(\fdepth{\gamma}, \max_{\delta \in \Delta}{\fdepth{\delta}})
  \end{mathpar}
\end{definition}

\begin{definition}\label{def:forcing} Given some \kl{Kripke structure}
\AP
  $\mathcal{K}$, the \intro{forcing} relation $\intro*\eforces{w}{\phi}{e}$
  between a \kl{world} $w$, a \kl{flower} $\phi$ and a \evaluation{w} $e$ is defined
  by induction on $\fdepth{\phi}$ as follows:
  \begin{description}
    \item[Atom] $\eforces{w}{p(\tvec{x})}{e}$ iff $\interp{\tvec{x}}_e
    \in \interp{p}_w$;
    \item[Flower]
    $\eforces{w}{\flower{\garden{\bx}{\Phi}}{\fset{i}{n}{\garden{\bx_i}{\Phi_i}}}}{e}$
    iff for every $w' \revaccess w$ and every \evaluation{w'} $e'$, if
    $\eforces{w'}{\Phi}{\update{e}{\bx}{e'}} $ then there is some $1 \leq i \leq
    n$ and \evaluation{w'} $e''$ such that
    $\eforces{w'}{\Phi_i}{\update{\update{e}{\bx}{e'}}{\bx_i}{e''}}$.
    \item[Bouquet] $\eforces{w}{\Phi}{e}$ iff $\eforces{w}{\phi}{e}$
    for every $\phi \in \Phi$.
  \end{description}
  % , and $w \forces \Phi$ if $\eforces{w}{\Phi}{e}$ for every
  % \evaluation{w} $e$.
\end{definition}

Lastly, we define the notion of \emph{semantic entailment} $\Phi \kentails{}
\Psi$ on \kl{bouquets}, mirroring the syntactic entailment $\Phi \entails{}
\Psi$ of the last section:

\begin{definition}
\AP
  Let $\mathcal{K}$ be a \kl{Kripke structure}, and $\Phi, \Psi$ some
  \kl{bouquets}. We say that $\Phi$ \intro{semantically entails} $\Psi$ in
  $\mathcal{K}$, written $\Phi \intro*\kentails{\mathcal{K}} \Psi$, when
  $\eforces{w}{\Phi}{e}$ implies $\eforces{w}{\Psi}{e}$ for every \kl{world} $w \in
  W$ and \evaluation{w} $e$. This entailment is \intro{valid} if it holds for any
  \kl{Kripke structure} $\mathcal{K}$, and in that case we simply write $\Phi
  \kentails{} \Psi$. We say that $\Phi$ is \intro{semantically equivalent} to
  $\Psi$, written $\Phi \intro*\kequiv{} \Psi$, when $\Phi \kentails{} \Psi$ and
  $\Psi \kentails{} \Phi$.
\end{definition}

\section{Completeness}\label{sec:Completeness}

% We now give a direct completeness proof for the natural fragment $\Nature$ of
% the \kl{flower calculus}: every true flower $\phi$ is naturally provable, i.e.
% $\kentails \phi$ implies $\entails{\Nature} \phi$. Since this fragment is
% analytic, we cannot adapt directly most of the completeness proofs for standard
% proof systems that can be found in the literature. Indeed, most of them exploit
% the transitivity of syntactic entailment $\entails{}$, and more precisely the fact
% that it is easily shown syntactically with the help of a non-analytic principle
% for composing proofs: in Hilbert systems it is the rule of \emph{modus ponens},
% in sequent calculi the \emph{cut} rule, and in natural deduction the
% \emph{substitution} theorem.

% Fortunately, a few people have noticed that with Kripke semantics, it is not too
% difficult to find completeness proofs that do not rely on the assumption of
% transitivity for $\entails{}$, thus allowing for a \emph{semantic} proof of
% cut-elimination
% \cite{10.1007/978-3-642-02261-6_17}\cite{ilik:tel-00529021}\cite{hutchison_semantic_2005}.
% Here we propose an adaptation of this technique to the \kl{flower calculus}, based on
% a completeness proof for cut-free sequent calculus given by Hermant in
% \cite{hutchison_semantic_2005}, which is itself close to the original
% completeness proof of Gödel with respect to classical Tarski models. A novelty
% of our proof is that it dispenses completely with the need for \emph{Henkin
% witnesses}.

We now outline a direct completeness proof for the \kl{natural} fragment
$\Nature$ of the \kl{flower calculus}: every true \kl{flower} $\phi$ is
naturally provable, i.e. $\kentails{} \phi$ implies $\entails{\Nature} \phi$. Since
this fragment is \kl{analytic}, we cannot reuse most completeness proofs from
the literature, because they usually rely on a non-\kl{analytic} principle like
the cut rule of sequent calculus. Our insight was to adapt techniques from the
\emph{semantic cut-elimination} proof given by Hermant in
\cite{hutchison_semantic_2005}, which is nonetheless relatively close to the
original completeness proof of Gödel. A novelty of our proof is that it
dispenses completely with the need for \emph{Henkin witnesses}.

% \subsection{Theories}

\AP
First we need to generalize our notions of syntactic and semantic entailment to
possibly \emph{infinite} sets of \kl{flowers}, so-called \emph{theories}:

\begin{definition}\label{def:theory}
  Any set $\tT \subseteq \flowers$ of \kl{flowers} is called a \intro{theory}.
  In particular, a \kl{bouquet} can be regarded as a finite \kl{theory}, by
  forgetting the number of repetitions of its elements. We say that a
  \kl{bouquet} $\Phi$ is \reintro{provable from} a \kl{theory} $\tT$, written
  $\tT \reintro*\entails{} \Phi$, if there exists a \kl{bouquet} $\Psi \subseteq
  \tT$ such that $\Psi \entails{} \Phi$. Given a \kl{Kripke structure}
  $\mathcal{K}$, a \kl{world} $w$ in $\mathcal{K}$ and a \evaluation{w} $e$, we
  say that $\tT$ is \reintro{forced} by $w$ under $e$, written
  $\reintro*\eforces{w}{\tT}{e}$, if $\eforces{w}{\phi}{e}$ for all $\phi \in
  \tT$. Then $\Phi$ is a \reintro{consequence} of $\tT$, written $\tT
  \reintro*\kentails{\mathcal{K}} \Phi$, if $\eforces{w}{\tT}{e}$ implies
  $\eforces{w}{\Phi}{e}$ for every \kl{world} $w$ in $\mathcal{K}$ and
  \evaluation{w} $e$.
\end{definition}

% The following notions are crucial to define the \emph{completion} procedure at
% the heart of any Gödel-style completeness proof\footnote{Neither the
% construction nor the use of the completion procedure in the completeness
% argument are detailed in this section. See Appendix
% \ref{app:Proofs-completeness}. }:

\begin{definition}
\AP
  A \kl{theory} $\tT$ is said to be \intro*\consistent{\psi} when $\tT
  \nentails{\Nature} \psi$, and \intro*\complete{\psi} when for all $\phi \in
  \flowers$, either $\tT, \phi \entails{\Nature} \psi$ or $\phi \in \tT$.
\end{definition}

Intuitively, a theory $\tT$ is \consistent{\psi} when one cannot deduce $\psi$
from it, and \complete{\psi} when it \emph{decides} any formula $\phi$
relatively to $\psi$. This is better understood by considering the special case
where $\psi = (\flower{{}}{{}})$ is the \emph{absurd} flower: then consistency
means that one cannot derive any contradiction from $\tT$; and completeness that
$\tT$ either \emph{refutes} $\phi$ syntactically with a proof of $\flower{\Phi,
\phi}{(\flower{{}}{{}})}$ for some $\Phi \subseteq \tT$, or already \emph{validates}
it ``semantically'', i.e. without the need for a proof since $\phi \in \tT$.

% \subsection{Adequacy}

The next two propositions constitute the central argument that allows the
completeness proof to go through despite the \kl{analyticity} of
$\Nature$-rules. They are a direct adaptation of \cite[Proposition
7]{hutchison_semantic_2005}, which Hermant identifies as ``an important property
of any $A$-consistent, $A$-complete theory, [...] that it enjoys some form of
the subformula property''.

Roughly, the first proposition captures the intuitionistic truth-conditions that
make a flower \emph{valid} (i.e. true in every model) by modelling them on
material implication, just like Peirce would do with his \kl{scroll} (see
Section \ref{sec:Flowers}): $\phi$ is true if the content $\Phi_i$ of one of its
petals (consequents) is, or if the content $\Phi$ of its pistil (antecedent) is
not.

\begin{proposition}[Analytic truth]\label{prop:inv-elem-flower}
  
  Let $\psi \in \flowers$, $\tT$ some \consistent{\psi} and \complete{\psi}
  \kl{theory}, and $\phi = \flower{\garden{\bx}{\Phi}}{\Delta}$ with $\Delta =
  \fset{i}{n}{\delta_i} = \fset{i}{n}{\garden{\bx_i}{\Phi_i}}$ such that $\phi
  \in \tT$. Then for every \kl{substitution} $\sigma : \bx$, either $\sigma(\Phi_i)
  \subseteq \tT$ for some $1 \leq i \leq n$, or $\tT \nentails{\Nature}
  \sigma(\Phi)$.
\end{proposition}
\begin{proof}
  Suppose the contrary, i.e. there is a \kl{substitution} $\sigma$ such that
  $\tT \entails{\Nature} \sigma(\Phi)$ and for all $1 \leq i \leq n$, there is
  some \Hyp{$\phi_i \in \Phi_i$}~{\HOne} such that $\sigma(\phi_i) \not\in \tT$.
  Thus by \kl[\complete]{$\psi$-completeness} of $\tT$, we get $\tT,
  \sigma(\phi_i) \entails{\Nature} \psi$. So there are $\Psi \subseteq \tT$ and
  $\Psi_i \subseteq \tT \cup \sigma(\phi_i)$ such that \Hyp{$\Psi
  \entails{\Nature} \sigma(\Phi)$}~{\HTwo} and \Hyp{$\Psi_i \entails{\Nature}
  \psi$}~{\HThree}. Now it cannot be the case that $\Psi_i \subseteq \tT$,
  otherwise by weakening and \kl[\consistent]{$\psi$-consistency} of $\tT$ we
  would have $\Psi_i \nentails{\Nature} \psi$. So there must exist $\Psi_i'
  \subseteq \tT$ such that \Hyp{$\Psi_i = \Psi_i' \cup
  \sigma(\phi_i)$}~{\HFour}. Again by weakening and
  \kl[\consistent]{$\psi$-consistency} of $\tT$, we get $\Psi, \bigcup_{i =
  1}^{n}{\Psi_i'}, \phi \nentails{\Nature} \psi$. Now we derive a contradiction
  by showing $\Psi, \bigcup_{i = 1}^{n}{\Psi_i'}, \phi \entails{\Nature} \psi$.
  Let $\cXi$ be a \kl{context} such that \Hyp{$\chyp{\Psi, \bigcup_{i =
  1}^{n}{\Psi_i'}, \phi}{\cXi}$}~{\HFive}. Then $\cfill{\cXi}{\psi}
  \steps{\Nature} \cfill{\cXi}{}$ with the following \kl{derivation}:
  $$
  \begin{array}{rlll}
    \cfill{\cXi}{\psi}
    &\step{\kl{epis}} &\cfill{\cXi}{\flower{\garden{{}}{{}}}{\garden{{}}{\psi}}} & \\
    &\step{\kl{poll{\ua}}} &\cfill{\cXi}{\flower{\garden{{}}{\phi}}{\garden{{}}{\psi}}} &\text{(\HFive)} \\
    &\step{\kl{ipis}} &\cfill{\cXi}{\flower{\garden{{}}{(\flower{\garden{{}}{\sigma(\Phi)}}{\sigma(\Delta)}), \phi}}{\garden{{}}{\psi}}} & \\
    &\step{\kl{poll{\da}}} &\cfill{\cXi}{\flower{\garden{{}}{(\flower{\garden{}{{}}}{\sigma(\Delta)}), \phi}}{\garden{{}}{\psi}}} &\text{(\HTwo, \HFive)} \\
    &\step{\kl{srep}} &\cfill{\cXi}{\flower{\garden{{}}{\phi}}{\garden{{}}{\fset{i}{n}{\flower{\sigma(\delta_i)}{\garden{{}}{\psi}}}}}} & \\
    &= &\cfill{\cXi}{\flower{\garden{{}}{\phi}}{\garden{{}}{\fset{i}{n}{\flower{\garden{\bx_i}{\sigma(\Phi_i)}}{\garden{{}}{\psi}}}}}} & \\
    &\nsteps{n}{\kl{poll{\da}}} &\cfill{\cXi}{\flower{\garden{{}}{\phi}}{\garden{{}}{\fset{i}{n}{\flower{\garden{\bx_i}{\sigma(\Phi_i)}}{\garden{{}}{{}}}}}}} &\text{(\HOne, \HThree, \HFour, \HFive)} \\
    &\nsteps{n}{\kl{epet}} &\cfill{\cXi}{\flower{\garden{{}}{\phi}}{\garden{{}}{{}}}} & \\
    &\step{\kl{epet}} &\cfill{\cXi}{} &
    % &\steps{\mathsf{ipis}} &\cfill{\cXi}{\flower{\garden{{}}{\phi}}{\garden{{}}{(\flower{\garden{}{\sigma_1 \circ \sigma(\Phi_1)}}{\garden{{}}{\psi}}), (\flower{\garden{\bx_1}{\sigma(\Phi_1)}}{\garden{{}}{\psi}}), \ldots, \\
    % && \qquad\quad~~\, (\flower{\garden{{}}{\sigma_n \circ \sigma(\Phi_n)}}{\garden{{}}{\psi}}), (\flower{\garden{\bx_n}{\sigma(\Phi_n)}}{\garden{{}}{\psi}})}}} \\
    % &\steps{\mathsf{poll{\da}}} &\cfill{\cXi}{\flower{\garden{{}}{\phi}}{\garden{{}}{(\flower{\garden{}{\sigma_1 \circ \sigma(\Phi_1)}}{\garden{{}}{}}), (\flower{\garden{\bx_1}{\sigma(\Phi_1)}}{\garden{{}}{\psi}}), \ldots, \\
    % && \qquad\quad~~\, (\flower{\garden{{}}{\sigma_n \circ \sigma(\Phi_n)}}{\garden{{}}{{}}}), (\flower{\garden{\bx_n}{\sigma(\Phi_n)}}{\garden{{}}{\psi}})}}} \\
  \end{array}
  $$
\end{proof}

Dually, the second proposition captures the grounds on which a flower can be
deemed \emph{invalid} (i.e. false in at least one model): $\phi$ is not true if
assuming that its pistil $\Phi$ is true is not sufficient to conclude that one
of its petals $\Phi_i$ is.

\begin{proposition}[Analytic refutation]\label{prop:inv-deriv-flower}
  
  Let $\psi \in \flowers$, $\tT$ some \consistent{\psi} and \complete{\psi}
  \kl{theory}, and $\phi = \flower{\garden{\bx}{\Phi}}{\Delta}$ with $\Delta =
  \fset{i}{n}{\delta_i} = \fset{i}{n}{\garden{\bx_i}{\Phi_i}}$ such that $\tT
  \nentails{\Nature} \phi$. Then for every $1 \leq i \leq n$ and \kl{substitution}
  $\sigma : \bx_i$, there is some $\phi_i \in \Phi_i$ such that $\tT, \Phi
  \nentails{\Nature} \sigma(\phi_i)$.
\end{proposition}
\begin{proof}
  Suppose the contrary, i.e. there are some $1 \leq i \leq n$ and $\sigma :
  \bx_i$ such that $\tT, \Phi \entails{\Nature} \sigma(\Phi_i)$.
  Therefore there must exist $\Psi \subseteq \tT$ and \Hyp{$\Phi_0 \subseteq
  \Phi$}~{\HOne} such that \Hyp{$\Psi, \Phi_0 \entails{\Nature}
  \sigma(\Phi_i)$}~{\HTwo}. By hypothesis, for every $\Phi' \subseteq \tT$ there
  is a \kl{context} $\cXi$ such that $\chyp{\Phi'}{\cXi}$ and $\cfill{\cXi}{\phi}
  \notsteps{\Nature} \cfill{\cXi}{}$. We now derive a contradiction
  by showing $\cfill{\cXi}{\phi} \steps{\Nature} \cfill{\cXi}{}$ for
  all $\cXi$ such that \Hyp{$\chyp{\Psi}{\cXi}$}~{\HThree}:
  $$
  \begin{array}{rlll}
    \cfill{\cXi}{\phi}
    &\step{\kl{ipet}} &\cfill{\cXi}{\flower{\garden{\bx}{\Phi}}{\garden{{}}{\sigma(\Phi_i)} \sep \Delta}} & \\
    &\step{\kl{poll{\da}}} &\cfill{\cXi}{\flower{\garden{\bx}{\Phi}}{\garden{{}}{{}} \sep \Delta}} &\text{(\HOne, \HTwo, \HThree)} \\
    &\step{\kl{epet}} &\cfill{\cXi}{} &
  \end{array}
  $$
\end{proof}

% \begin{remark}
% Note that both \ref{prop:inv-elem-flower} and \ref{prop:inv-deriv-flower} are
% proved indirectly by contradiction, but that the contradiction is obtained by
% exhibiting a concrete derivation, exploiting respectively the universal
% instantiation of pistils with the \kl{ipis} rule, and the existential
% instantiation of petals with the \kl{ipet} rule. Thus there seems to be some
% constructive content to these proofs, despite their use of a classical reasoning
% principle. Also, this is the only place in the completeness proof where we build
% derivations, and all $\Nature$-rules (and only them) seem to be required to
% conclude.
% \end{remark}

Next, we define the so-called \emph{universal Kripke structure} $\rewolF{\psi}$
relative to a \kl{flower} $\psi$:
% , which satisfies the property that $\Phi
% \kentails{\rewolF{\psi}} \Psi$ implies $\Phi \kentails{} \Psi$ for any
% $\psi$\sidenote{We do not prove this property here, but this is a standard
% result in the literature, and the fact that we work on flowers instead of
% formulas should not have any significance.}.

\begin{definition}\label{def:kcanon}
\AP
  Let $\psi \in \flowers$. The \intro{universal Kripke structure}
  $\intro*\rewolF{\psi}$ has:
  \begin{itemize}
    \item The set of \consistent{\psi} and \complete{\psi} \kl{theories} as its \kl{worlds};
    \item Set inclusion $\subseteq$ as its \kl{accessibility} relation;
    \item For each \kl{world} $\tT$, a \kl{first-order structure} whose \kl{domain}
      is the set of variables $\vars$, and whose \kl{interpretation} map is
      given by $\interp{p}_\tT = \compr{\tvec{x}}{p(\tvec{x}) \in \tT}$.
  \end{itemize}
  One can easily check that the monotonicity conditions of \kl{Kripke
  structures} hold for $\rewolF{\psi}$.
\end{definition}

We are now equipped to formulate the main \emph{adequacy} lemma, which relates
\kl{forcing} in $\rewolF{\psi}$ to \kl[\consistent]{$\psi$-consistency} and
\kl[\complete]{$\psi$-completeness}:
% thanks to Propositions \ref{prop:inv-elem-flower} and
% \ref{prop:inv-deriv-flower}:

\begin{lemma}[Adequacy]\label{lem:adequacy}

  Let $\phi, \psi \in \flowers$, $\tT$ a \consistent{\psi} and \complete{\psi}
  \kl{theory}, and $\sigma$ a \kl{substitution}.
  % such that $\fv(\sigma(x)) \cap \bv(\phi) = \emptyset$ for any variable $x$.
  Then
  \begin{enumerate*}
    \item $\sigma(\phi) \in \tT$ implies $\eforces{\tT}{\phi}{\sigma}$, and
    \item $\tT \nentails{\Nature} \sigma(\phi)$ implies $\neforces{\tT}{\phi}{\sigma}$
  \end{enumerate*}.
\end{lemma}
\begin{proof}
  The proof goes by induction on $\fdepth{\phi}$. We only give an informal
  sketch, see Appendix \ref{app:Proofs-completeness} for the detailed proof.
  There are just two cases to consider:
  \begin{description}
    \item[Base case] $\phi = p(\tvec{x})$. The first statement is trivial. The
    second statement is immediate from reflexivity and weakening lemmas on the
    hypothetical provability relation ${\entails{}}$.
    \item[Recursive case] $\phi = \flower{\gamma}{\Delta}$. The first statement
    follows from Proposition \ref{prop:inv-elem-flower}. The second statement
    follows from Proposition \ref{prop:inv-deriv-flower}, as well as the
    existence and properties of the completion procedure.
  \end{description}
\end{proof}

We get the completeness theorem as a near-direct consequence:

\begin{theorem}[Completeness]\label{thm:flowers-completeness}
  $\Phi \kentails{} \Psi$ implies $\Phi \entails{\Nature} \Psi$.
\end{theorem}
% \begin{proof}
%   Let $\tT$ be a \consistent{\psi} \kl{theory}. We prove that $\tT \notkentails{} \psi$ by
%   showing in particular that $\tT \notkentails{\rewolF{\psi}} \psi$, and more
%   specifically that $\eforces{\completion{\tT}}{\tT}{\idsubst}$ but
%   $\neforces{\completion{\tT}}{\psi}{\idsubst}$. Then it follows by (classical)
%   contraposition that $\tT \kentails{} \psi$ implies $\tT \entails{\Nature} \psi$ for
%   any $\psi$ and any $\tT$, and thus we can conclude.
%   \begin{itemize}
%     \item Let $\phi \in \tT$. Then $\idsubst(\phi) = \phi \in \completion{\tT}$,
%     thus by \kl[\consistent]{$\psi$-consistency} and \kl[\complete]{$\psi$-completeness} of the
%     completion (Lemma \ref{lem:completion-consistent-complete}), one can apply
%     adequacy (Lemma \ref{lem:adequacy}) to get
%     $\eforces{\completion{\tT}}{\phi}{\idsubst}$.
%     \item Similarly, we can apply adequacy (Lemma \ref{lem:adequacy}) to get
%     $\neforces{\completion{\tT}}{\psi}{\idsubst}$.
%   \end{itemize}
% \end{proof}

Combined with strong deduction (Theorem \ref{thm:flowers-strong-deduction}),
this also yields a strong completeness theorem for the full \kl{flower
calculus}\footnote{Actually it already works for the fragment $\Nature \cup
\{\kl{grow}\}$, thanks to the proof of the strong deduction theorem (see
Appendix \ref{app:Proofs-deduction}).}:

\begin{corollary}[Strong completeness]\label{thm:flowers-strong-completeness}
  $\Phi \kentails{} \Psi$ implies $\Psi \steps{} \Phi$.
\end{corollary}
% \begin{proof}
%   By \ref{lem:completeness-contra} we have that $\flower{\Phi}{\Psi} \steps{}
%   \emptyset$, and by \ref{thm:flowers-strong-deduction} we can conclude.
% \end{proof}

Finally, the composition of the soundness, completeness and deduction theorems
(\ref{thm:flowers-soundness}, \ref{thm:flowers-completeness} and
\ref{thm:flowers-deduction}) gives the admissibility of $\Culture$-rules, and
thus the \kl{analyticity} of the \kl{flower calculus}:

\begin{corollary}[Cult-elimination]
  If $\Phi \entails{} \Psi$ then $\Phi \entails{\Nature} \Psi$.
\end{corollary}

\section{Related works}\label{sec:Conclusion}

\subparagraph*{Intuitionistic \kl{EGs}}

We have already mentioned the seminal work of Oostra, who introduced in
\cite{oostra_graficos_2010} an intuitionistic version of \kl{Alpha}. In
\cite{oostra_graficos_2011} he describes its natural extension with \kl{LoIs} to get
an intuitionistic version of \kl{Beta}, and in
\cite{10.1007/978-3-030-86062-2_16} he gives formal soundness and completeness
proofs for intuitionistic \kl{Alpha}, based on a linear notation for \kl{graphs}.
Ma and Pietarinen have developed in \cite{minghui_graphical_2019} their own
system of intuitionistic \kl{EGs} for propositional logic, with a different set
of inference rules than Oostra's. They give a more systematic proof theory,
including deduction, soundness and completeness theorems with respect to Heyting
algebras.

Our work brings several new contributions on top of those:
\begin{description}
  \item[Variadicity] Our multiset-based definition of \kl{flowers} captures
  faithfully the \emph{variadic} nature of \kl{juxtaposition} and $n$-ary
  \kl{scrolls} in the diagrammatic syntax. In contrast, previous formalizations
  rely on a restricted inductive syntax which only captures \kl{graphs} that are
  isomorphic to formulas built with binary connectives.
  % \footnote{Thus we reject
  % the claim made in \cite{minghui_graphical_2019} that their system is ``solving
  % the problem of defining a sequent calculus in the style of deep inference for
  % intuitionistic propositional logic''. In our opinion, \kl{flowers} are closer to a
  % form of nested sequent, although there is no consensus in the literature on
  % what makes some inductive data structure a nested sequent.}.

  \item[Intuitionistic binders] While replacing \kl{LoIs} with \kl{binders}
  and variables has already been done by Sowa in the context of classical \kl{EGs}
  \cite{sowa_peirces_2011}, it seems like we are the first to adapt the idea to
  the intuitionistic setting.

  \item[Analyticity] To our knowledge, we are the first to give a Kripke
  semantics to a syntax based on \kl{EGs}, and to use this to obtain an
  \kl{analyticity} result\footnote{Ma and Pietarinen claim in
  \cite{ma_proof_2017} that \kl{Alpha} is analytic because it can simulate the
  cut rule of sequent calculus. This is a misinterpretation, since this supports
  precisely the \emph{contrary}: the ability to simulate the cut rule with a
  constant number of rules implies the \emph{non-analyticity} of one the rules
  involved (namely, Peirce's \textsf{Deletion} rule). Still, the notion of
  analyticity is not yet fully understood in deep inference systems, as
  discussed in \cite{bruscoli_analyticity_2019}.}.
  
  \item[Invertibility] The \kl{natural} fragment of the flower calculus appears
  to be the first proof system based on \kl{EGs} where all rules are
  \emph{invertible}.
\end{description}

\subparagraph*{Deep inference}

While the deep inference literature is most furnished with systems for classical
logic, a few works tackle intuitionistic logics: the seminal work of Tiu, who
proposed a \emph{calculus of structures} for intuitionistic \kl{FOL}
\cite{tiu_local_2006}, was followed by computational interpretations of the
implicative fragment in Guenot's thesis \cite{guenot_nested_2013}. There are
also \kl{nested sequent} systems for (propositional) full intuitionistic linear
logic \cite{clouston-annotation-free-2013}, standard and constant-domain
intuitionistic \kl{FOL} \cite{fitting-nested-2014}, and intuitionistic modal
logics
\cite{Chaudhuri2016ModularFP,kuznets_maehara-style_2019,lyon_nested_2021}. The
\kl{flower calculus} is closer to Guenot's \kl{nested sequent} calculi for
implicative logic which also function as rewriting systems, but extends them
to full intuitionistic \kl{FOL}.

\subparagraph*{Labelled sequent calculi}

For a long time, it was believed that there could not be fully invertible proof
systems for intuitionistic logics, even in the propositional case. While this
might be true in standard Gentzen formalisms, recent works have shown that it is
possible in the context of \emph{labelled sequent calculi}: first with Lyon's
\sys{G3IntQ} calculus for \kl{FOL} \cite[Section~3.3]{lyon_refining_2021}, and then
with the calculus \sys{labIS4_\leq} of Girlando et al.\ for the modal logic S4
\cite{girlando_intuitionistic_2023}. In these systems, invertibility is made
possible by the addition of \emph{semantic} information to sequents, in the form
of so-called \emph{labels} and \emph{relational atoms} that respectively encode
the \kl{worlds} and \kl{accessibility} relations of \kl{Kripke structures}. The
\kl{flower calculus} follows instead a purely \emph{syntactic} approach, by
relying on deep inference to retrieve what would normally be semantic
information from the \kl{context} $\cXi$ in the pollination rules \kl{poll{\ua}}
and \kl{poll{\da}}.

\subparagraph*{Categorical \kl{EGs}}

\AP
Since the seminal work of Brady and Trimble in 2000 on the formalization of
\kl{EGs} in category theory \cite{brady_categorical_2000,brady_string_nodate},
there have been various efforts to find rich categorical axiomatizations of
\kl{Beta}. The first approach --- initiated in \cite{brady_string_nodate} --- is
based on \emph{string diagrams}, and has recently enabled strong connections
with \emph{Frobenius algebras} and
\emph{bicategories or relations}
\cite{mellies_bifibrational_2016,pietarinen_compositional_2020,bonchi_diagrammatic_2024}.
A second approach makes use of the concept of \emph{generic figure}
\cite{caterina_new_2020}, introduced by Reyes as a basic building block for
\intro{topos theory} \cite{generic_figures}. We do not know however of any
attempt to uncover the categorical structures underlying intuitionistic
\kl{EGs}. The \kl{flower calculus} might be an interesting candidate, in that
the invertibility of the \kl{natural} fragment could enable a purely
\emph{equational} approach.

\subparagraph*{Coherent logic}

We noticed a formal connection between \kl{flowers} and \intro{coherent logic},
a subset of the formulas of \kl{FOL} discovered by Skolem in 1920
\cite{skolem_geometric} that is capable of expressing many mathematical
theories, and has close connections to \kl{topos theory}
\cite[Section~D3.3]{Johnstone2002-rm}. Indeed, the interpretation
$\toform{\flower{\garden{\bx}{\Phi}}{\Delta}}$ of a generic \kl{flower} is given
by the following formula, which has exactly the shape of a coherent formula as
described e.g.\ in \cite{bezem_automating_2005}:
\begin{mathpar}
  \forall \bx. \left(\bigwedge_{\phi \in \Phi} \toform{\phi} \Impl \bigvee_{\garden{\by}{\Psi} \in \Delta}\exists \by. \bigwedge_{\psi \in \Psi} \toform{\psi}\right)
\end{mathpar}
The only difference is that \kl{flowers} can be \emph{nested}, while coherent
formulas (also called coherent \emph{sequents}) are \kl{first-order}, in the
sense that $\phi$ and $\psi$ must be atoms. Coherent formulas appear in the
theory of \emph{focusing} in sequent calculi \cite{marin_axioms_2022}, and they
lend themselves to simple proof search procedures that allow for
\emph{explainable proof automation} in \kl{ITPs}
\cite{bezem_automating_2005,narboux_larus}. A \kl{higher-order} variant of
coherent formulas that is almost isomorphic to \kl{flowers} has also been used
to construct an intuitionistic version of the \emph{arithmetical hierarchy}, as
well as a fully \emph{non-invertible} proof system for propositional
intuitionistic logic \cite{brock-nannestad_intuitionistic_2019}.

\subparagraph*{Graph calculi}

In the last twenty years, Veloso et al. have studied a series of so-called
\intro{graph calculi}, where first-order relations are represented by graphs in
the sense of graph theory, and inference rules as graph transformations. The
first \kl{graph calculus} was introduced informally by Curtis and Lowe in 1996
\cite{curtis_proofs_1996}, as a graphical notation supposedly capturing both
relational calculus, and the sequential calculus of Karger and Hoare
\cite{von_karger_sequential_1995}. Veloso et al. gave sound and complete syntax
and semantics to the calculus in \cite{de_freitas_graph_2009}, showing that it
captures \emph{positive} first-order logic on \emph{binary} relations. They then
extended their formalism to support relational complementation (negation)
\cite{goel_calculus_2010} as well as various modal
\cite{veloso_graph_2017,veloso_graph_2015} and dynamic logics
\cite{veloso_pdl_2014}.

Graph calculi only handle \emph{binary} relations and classical logic, while
\kl{EGs} and the flower calculus support relations of arbitrary arity and
intuitionistic logic. We conjecture that the relationship between graph calculi
and \kl{EGs} is similar to that between \emph{commutative diagrams} and
\emph{string diagrams} in category theory: the former represent
relations/morphisms as edges between individuals/objects, while the latter
represent them dually as points related by lines. \kl{EGs} could then be
understood as a \emph{hypergraph} generalization of graph calculi, where
\kl{lines of identity} are hyperedges connecting multiple predicate vertices.

\subparagraph*{Development calculi}

Through their backward reading, the rules of the \kl{flower calculus} can be
understood as primitive \emph{tactics} for building proofs interactively. In
\cite[Chapter~3]{ayers_thesis}, Ayers calls such systems \emph{development
calculi}. In particular, he presents his own development calculus inspired by
McBride's OLEG system \cite{mcbride_dependently_2000} and Ganesalingam \&
Gowers's prover \cite{ganesalingam_fully_2017} called the \texttt{Box} calculus,
where goals are represented by a so-called \texttt{Box} data structure very
similar to flowers. In particular, \texttt{Box}es have so-called
\emph{disjunctive pairs} to reduce backtracking, that correspond to the
\kl{petals} of \kl{flowers}. The main difference is that the \texttt{Box}
calculus is based on dependent type theory instead of \kl{FOL}: this allows to store
the partial proof terms inside of the \texttt{Box}es themselves, while this
information is lost during the construction of \kl{flowers}. However, there is
no completeness nor analyticity result for the \texttt{Box} calculus. It would
be interesting to investigate further connections, in order to develop a
dependently-typed version of the \kl{flower calculus}.

% Ayers also mentions the category-theoretical treatment of development calculi by
% Sterling and Harper \cite{sterling_algebraic_2017}, that abstracts from any
% particular type of judgment. Thus it might be possible to fit the flower
% calculus into this framework, by identifying the set of flowers $\flowers$ as a
% category of \emph{nested judgments}\footnote{Nested judgments are already
% considered in some recent categorical semantics of type theory, and in
% particular those in Sterling's thesis \cite{Sterling2022}. See also
% \cite{huang2023synthetic} for a (technical) introduction to the subject.}.

\bibliography{main}

\begin{figure}
  \input{natural-graphical.tex}
  \caption{Graphical presentation of \kl{natural} rules $\Nature$}
  \label{fig:natural-graphical}
\end{figure}

\begin{figure}
  \input{cultural-graphical.tex}
  \caption{Graphical presentation of \kl{cultural} rules $\Culture$}
  \label{fig:cultural-graphical}
\end{figure}
% \subsection{Graphical rules}\label{app:Graphical}

\begin{figure}
  \captionsetup[subfigure]{justification=centering}
  \centering
  \begin{subfigure}{\textwidth}
    \centering
    \input{proof-example-textual.tex}
    \caption{Textual presentation}
  \end{subfigure}
  \par\bigskip
  \begin{subfigure}{\textwidth}
    % \centering
    \input{proof-example-graphical.tex}
    \caption{Graphical presentation}
  \end{subfigure}
  \caption{A \kl{natural} proof in the flower calculus}
  \label{fig:flowers-proof-example}
\end{figure}

\newpage
\appendix

\section{Comparison with \kl{EGs}}\label{app:Comparison}

In this section, we give a detailed comparison of the rules of the \kl{flower
calculus} with the \kl{illative transformations} of Peirce's system \kl{Beta}.
To ease the presentation, we introduce an inductive syntax for the graphs of
\kl{Beta} based on our notion of \kl{garden}.

\begin{definition}[$\beta$-graph]\label{def:beta-graph}
  
  The sets of \intro{$\beta$-nodes} $\intro*\bnodes$, \intro{$\beta$-gardens}
  $\intro*\bgardens$ and \intro{$\beta$-graphs} $\intro*\bgraphs$ are defined
  mutually inductively:
  \begin{description}
    \item[Atom] If $p \in \psymbs$ and $\tvec{x} \in \vars^{\arity(p)}$, then
    $p(\tvec{x}) \in \bnodes$;
    \item[Graph] If $G \subset \bnodes$ is a finite multiset, then $G
    \in \bgraphs$.
    \itemAP[Garden] If $\bx \subset \vars$ is a finite set and $G
    \subset \bnodes$ a finite multiset, then $\intro*\bgarden{\bx}{G} \in
    \bgardens$;
    \itemAP[Cut] If $\gamma \in \bgardens$, then $\intro*\bcut{\gamma} \in
    \bnodes$.
  \end{description}
  % An \emph{$\alpha$-graph} is a multiset of atomic propositions and
  % $\alpha$-graphs.
\end{definition}

\begin{remark}
  Note that a \kl{$\beta$-graph} is defined as a multiset of \kl{$\beta$-nodes},
  just like a \kl{bouquet} is a multiset of \kl{flowers}. Then like in the
  \kl{flower calculus}, and unlike in Peirce's original system, \kl{binders}
  (\kl{LoIs}) cannot appear at the top-level of $\SA$.
\end{remark}

\begin{example}
  The lower graph in Figure~\ref{fig:loi} can be written in textual notation as
  the expression $\bcut{\bgarden{x}{R(x),\bcut{\bgarden{{}}{P(x)}}}}$.
\end{example}

\AP
We also adapt the notion of \kl{context} to the \kl{Beta} setting:

\begin{definition}[$\beta$-context]\label{def:beta-contexts}
  \intro{$\beta$-contexts} $\cG$ are defined inductively by the following
  grammar:
  \begin{mathpar}
    \cG, \cH, \cK \Coloneq G, \cg
    \and
    \cg, \ch, \ck \Coloneq \hole \mid [\cG]
  \end{mathpar}
\end{definition}

\begin{definition}[Polarity]
\AP The number of \emph{inversions} $\intro*\binv(\cG)$ of a
  \kl{$\beta$-context} $\cG$ is:
  \begin{mathpar}
    \binv(\hole) = 0
    \quad\,\,
    \binv(G, \cg) = \binv(\cg)
    \quad\,\,
    \binv([\cG]) = 1 + \binv(\cG)
  \end{mathpar}
  % We say that a \kl{$\beta$-context} $\cG$ is \intro(beta){positive} if
  % $\binv(\cG)$ is even, and \intro(beta){negative} otherwise. We denote
  % \kl(beta){positive} and \kl(beta){negative} \kl{$\beta$-contexts} respectively
  % by $\cGP$ and $\cGN$.
\end{definition}

The definitions of \kl{free variables}, \kl{bound variables} and
\kl{substitutions} can also be adapted straightforwardly.

The set of rules of \kl{Beta} is given in Figure~\ref{fig:beta}. \kl{Iter},
\kl{Deit}, \kl{Ins} and \kl{Del} correspond respectively to the Iteration,
Deiteration, Insertion and Deletion principles of \sys{Alpha}, while
\kl{Dcut{\da}} and \kl{Dcut{\ua}} capture the Double-cut principle. As for
\kl{Unif{\da}} and \kl{Unif{\ua}}, they correspond roughly to the principles of
Insertion and Deletion applied to \kl{LoIs}. Indeed in their bottom-up reading,
\kl{Unif{\ua}} and \kl{Unif{\da}} can be understood as capturing respectively
the operations of \emph{unification} and \emph{anti-unification} on two
variables:
\begin{description}
  \item[\textbf{Unification}] substituting $z$ for $y$ and removing the
\kl{binder} for $y$ is equivalent in purpose to adding a new \kl{LoI} between
the outermost point of the \kl{LoI} associated to $y$, and some point of the
\kl{LoI} associated to $z$ in the same area (which is assumed to exist by
well-scopedness);
  \item[\textbf{Anti-unification}] substituting $y$ for $z$ and adding a
  \kl{binder} for $y$ is equivalent in purpose to severing the \kl{LoI}
  associated to $z$ at the location where the \kl{binder} for $y$ is introduced.
\end{description}
There is no need to formulate an equivalent of the Iteration and Deiteration
principles for \kl{LoIs}. Indeed, their purpose is to manage \emph{locally} the
\emph{extension} of a line. With \kl{binders}, the notion of extension is
replaced with that of \emph{scope}, which is handled \emph{globally} and
automatically in the definition of \kl{substitutions}.

\begin{figure}
  \captionsetup[subfigure]{justification=centering}
  \centering
  \begin{subfigure}[b]{0.69\textwidth}
    \input{beta.tex}
    \caption{Rules of \kl{Beta}}
    \label{fig:beta}
  \end{subfigure}
  \begin{subfigure}[b]{0.3\textwidth}
    \centering
    \begin{mathpar}
      \R[\intro{epis{\da}}]
        {\Phi}
        {\flower{\garden{{}}{{}}}{\garden{{}}{\Phi}}}
    \end{mathpar}
    \caption{Converse of \kl{epis} rule}
    \label{fig:flowers-episda}
  \end{subfigure}
  \caption{Comparing rules}
\end{figure}

Let us now review the rules of the \kl{flower calculus} in more detail, starting
with the fragment that is a direct adaptation of the rules of \kl{Beta}:

\begin{description}
  \item[Blank Antecedant (\kl{epis})]

    It allows to enclose any \kl{bouquet} in a \kl{petal} attached to an \textsf{(e)}mpty
    \textsf{(pis)}til. This is a weaker, intuitionistic version of
    \kl{Dcut{\ua}}, that was already identified by Peirce as the ``collapsing of
    a scroll's walls' \cite[p.~534]{peirce_prolegomena_1906}, and is called the
    rule of Blank Antecedant in \cite{minghui_graphical_2019}. The converse rule
    that would correspond to \kl{Dcut{\da}} --- rule \kl{epis{\da}} in
    Figure~\ref{fig:flowers-episda} --- is actually shown to be
    \emph{admissible} by our completeness theorem\footnote{Interestingly,
    although an equivalent of \kl{epis{\da}} is included in
    \cite{minghui_graphical_2019}, there is no mention of it by Peirce in
    \cite{peirce_prolegomena_1906}. It may be a sign that Peirce already
    intuited its admissibility, or at least considered this direction of the
    transformation unworthy of attention.}.
  % \item[Blank Antecedant (\kl{epis{\ua}}, \kl{epis{\da}})]
  %   They allow to enclose (\kl{epis{\ua}}) or free (\kl{epis{\da}}) any
  %   bouquet in a petal attached to an \textsf{(e)}mpty \textsf{(pis)}til. They
  %   correspond to the rule \kl{BA} of \ref{fig:rule-empty-antecedant}, which is
  %   a weaker, intuitionistic version of the classical rules \kl{Dcut{\ua}} and
  %   \kl{Dcut{\da}} of \sys{Alpha} and \sys{Beta}. \kl{epis{\da}} will be shown
  %   to be admissible in \ref{sec:Completeness}, which might be related to the
  %   co-admissibility of \kl{Dcut{\da}} in \sys{Alpha} (Corollary
  %   \ref{cor:adm-ins}).

  \item[(De)iteration (\kl{poll{\da}}, \kl{poll{\ua}})]
    The \textsf{(poll)}ination rules \kl{poll{\da}} and \kl{poll{\ua}}
    correspond respectively to \kl{Iter} and \kl{Deit}, but reformulated with
    the \kl{pollination} relation (Definition~\ref{def:pollination}). In fact in
    their textual presentation (Figure~\ref{fig:flower-calculus}), they are more
    general than (de)iteration rules, because $\chyp{\Phi}{\cXi}$ allows the
    \kl{pollinating} \kl{bouquet} $\Phi$ to be \emph{scattered} in the \kl{context}
    $\cXi$, i.e. its \kl{flowers} need not be located in the same area. On the
    contrary in their graphical presentation
    (Figure~\ref{fig:natural-graphical}), they are less general since only one
    \kl{flower} can be \kl{pollinated} at a time, rather than an entire \kl{bouquet} of
    \kl{flowers} residing in the same area. But it is easy to see that all these
    variants are equivalent in deductive power, since the pollination of a
    \kl{bouquet} (however scattered) can always be simulated by the successive
    pollinations of each of its \kl{flowers}.

  \item[Insertion/Deletion (\kl{grow}, \kl{crop}, \kl{pull}, \kl{glue})]
    They correspond to \kl{Ins} and \kl{Del}, but have doubled in number to
    account for the syntactic distinction between \kl{pistils} and \kl{petals}. More
    precisely, rules \kl{grow} and \kl{crop} allow to insert and delete entire
    \kl{flowers}, while rules \kl{pull} and \kl{glue} deal with \kl{petals}. As for
    pollination rules, manipulating single \kl{flowers}/\kl{petals} (graphical version) or
    entire \kl{bouquets}/\kl{corollas} (textual version) does not change the deductive
    power of the rules.
    
  \item[Unification (\kl{ipis}, \kl{ipet}, \kl{apis}, \kl{apet})]
    Rules \kl{ipis} and \kl{ipet} allow to \textsf{(i)}nstantiate a \kl{sprinkler}
    located respectively in a \textsf{(pis)}til ($\forall$) and a
    \textsf{(pet)}al ($\exists$) with an arbitrary substitution, while rules
    \kl{apis} and \kl{apet} do the opposite operation of \textsf{(a)}bstracting
    a set of variables by introducing a \kl{sprinkler}. They correspond respectively
    to a generalization of \kl{Unif{\ua}} and \kl{Unif{\da}}, where the variable
    substitution $\subst{}{z}{y}$ becomes an arbitrary substitution $\sigma$.
    Once again, we have twice the amount of rules to account for the
    \kl{pistil}/\kl{petal} distinction, which is not surprising since in the \kl{LoI}
    syntax of \sys{EGs}, they are special cases of Insertion/Deletion. Note that
    for the instantiation rules \kl{ipis}/\kl{ipet} to be invertible, we
    duplicate the whole \kl{flower}/\kl{petal} where the \kl{sprinkler} occurs, mirroring what
    is done in multi-conclusion sequent calculi.
\end{description}

The last two rules mainly handle the behavior of disjunctive and absurd
statements, i.e. \kl{flowers} with respectively $n \geq 2$ and $n = 0$ \kl{petals}, and
are closer to sequent-style introduction/elimination rules:

\begin{description}
  \item[Disjunction Introduction (\kl{epet})]
    It allows to erase any \kl{flower} with an \textsf{(e)}mpty \textsf{(pet)}al.
    According to Oostra \cite[p.~109]{oostra_advances_2022}, Peirce already
    identified \kl{epet} as a component of his decision procedure for
    \sys{Alpha} (it is simply called ``Operation 1'' in
    \cite{oostra_advances_2022}). This is no coincidence, since we precisely
    came up with this rule when trying to design a decision procedure for
    \kl{flowers}.

  \item[Disjunction/Falsehood Elimination (\kl{srep})]
    It corresponds to a $n$-ary generalization of the left introduction rule for
    \kl{disjunction} in sequent calculus, the $0$-ary case capturing
    \kl{falsehood} elimination (\textit{ex falso quodlibet}) as illustrated in
    the proof of Figure~\ref{fig:flowers-proof-example}. The binary case is also
    used in the intuitionistic \kl{EGs} system of \cite{minghui_graphical_2019}
    together with its converse, which is also shown to be admissible by our
    completeness theorem. The name \kl{srep} is short for
    \textsf{(s)}elf-\textsf{(rep)}roduction, which is more clearly visualized in
    the graphical version of the rule in Figure~\ref{fig:natural-graphical}.
    % Through the Curry-Howard correspondence, it can be related to the
    % \emph{pattern-matching generator} found in modern editors of some functional
    % programming languages, such as the Hazel structure editor and the Agda proof
    % assistant \cite{yuan-live-2023}.
\end{description}

\section{Soundness}\label{app:Soundness}

In this section, we show that every rule of the \kl{flower calculus} is
\emph{sound} with respect to our Kripke semantics for \kl{flowers}, and thus
that $\entails{} \phi$ implies $\kentails{} \phi$ for every $\phi$. We start
with a few trivial facts about the \kl{update} operation of Definition
\ref{def:update}:

\begin{observation}[Associativity]\label{obs:update-assoc}
  $\update{\update{f}{R}{g}}{S}{h} = \update{f}{R \cup S}{(\update{g}{S}{h})}$.
\end{observation}

\begin{observation}[Commutativity]\label{obs:update-comm}
  If $R \cap S = \emptyset$ then $f \upd{R} g \upd{S} h = f \upd{S} h \upd{R}
  g$.
\end{observation}

\begin{observation}[Agreement]\label{obs:update-agree}
  If $f(x) = g(x)$ for all $x \in R$ then $h \upd{R} f = h \upd{R} g$.
\end{observation}

\begin{observation}[Idempotency]\label{obs:update-idempot}
  $f \upd{R} f = f$.
\end{observation}

Semantic entailment is obviously a reflexive and transitive relation:

\begin{observation}[Reflexivity]\label{obs:kentails-refl}
  $\Phi \kentails{} \Phi$.
\end{observation}

\begin{observation}[Transitivity]\label{obs:kentails-trans}
  If $\Phi \kentails{} \Psi$ and $\Psi \kentails{} \Xi$, then $\Phi \kentails{} \Xi$.
\end{observation}

The two following lemmas will be useful to reason on the \kl{forcing} relation
(Definition \ref{def:forcing}):

\begin{lemma}[Monotonicity]\label{lem:access-mono}
  If $w \access w'$ and $\eforces{w}{\phi}{e}$ then $\eforces{w'}{\phi}{e}$.
\end{lemma}
\begin{proof}
  By a straightforward induction on $\fdepth{\phi}$.
\end{proof}

\begin{lemma}[Mirroring]\label{lem:mirroring}
  
  $\eforces{w}{\sigma(\phi)}{e}$ iff
  $\eforces{w}{\phi}{\update{e}{\bx}{\interp{\sigma}_e}}$ for $\sigma : \bx$
  \kl{capture-avoiding} in $\phi$ and $\bx \cap \bv(\phi) = \emptyset$.
\end{lemma}
\begin{proof}
  By induction on $\fdepth{\phi}$.
  \begin{description}
    \item[Base case]
    \newcommand{\esigma}{\update{e}{\bx}{\interp{\sigma}_e}} Suppose $\phi =
    p(\tvec{y})$. We show that $\interp{\sigma(\tvec{y})}_e \in \interp{p}_w$
    iff $\interp{\tvec{y}}_{\esigma} \in \interp{p}_w$ by proving that
    $\interp{x}_{\esigma} = \interp{\sigma(x)}_e$ for any variable $x$. Either:
    \begin{itemize}
      \item $x \in \bx$, and $\interp{x}_{\esigma} = \interp{\sigma}_e(x) = \interp{\sigma(x)}_e$; or
      \item $x \not\in \bx$, and $\interp{x}_{\esigma} = e(x) = \interp{x}_e = \interp{\sigma(x)}_e$.
    \end{itemize}

    \item[Recursive case] Suppose $\phi =
    \flower{\garden{\by}{\Phi}}{\fset{i}{n}{\garden{\bz_i}{\Psi_i}}}$. We show
    that
    $\eforces{w}{\flower{\garden{\by}{\sigma(\Phi)}}{\fset{i}{n}{\garden{\bz_i}{\sigma(\Psi_i)}}}}{e}$
    implies
    $\eforces{w}{\flower{\garden{\by}{\Phi}}{\fset{i}{n}{\garden{\bz_i}{\Psi_i}}}}{\update{e}{\bx}{\interp{\sigma}}_e}$,
    the argument working in both directions. Let $w' \revaccess w$ and $e'$ a
    \evaluation{w'} such that
    $\eforces{w'}{\Phi}{\update{e}{\bx}{\update{\interp{\sigma}_e}{\by}{e'}}}$.
    Since $\sigma$ is \kl{capture-avoiding} in $\phi$, we know that $\fv(\sigma(x))
    \cap \by = \emptyset$, and thus $\interp{\sigma}_e(x) =
    \interp{\sigma(x)}_{e} = \interp{\sigma(x)}_{e \upd{\by} e'} =
    \interp{\sigma}_{e \upd{\by} e'}(x)$ for any $x \in \bx$. Hence by
    Observation \ref{obs:update-agree}
    $\eforces{w'}{\Phi}{\update{e}{\bx}{\update{\interp{\sigma}_{\update{e}{\by}{e'}}}{\by}{e'}}}$,
    and since by hypothesis $\bx \cap \by = \emptyset$ we obtain
    $\eforces{w'}{\Phi}{\update{e}{\by}{\update{e'}{\bx}{\interp{\sigma}_{\update{e}{\by}{e'}}}}}$
    by Observation \ref{obs:update-comm}. Then by IH we get
    $\eforces{w'}{\sigma(\Phi)}{\update{e}{\by}{e'}}$, and thus by hypothesis
    $\eforces{w'}{\sigma(\Psi_i)}{\update{e}{\by}{\update{e'}{\bz_i}{e''}}}$ for
    some $1 \leq i \leq n$ and \evaluation{w'} $e''$. Again by IH we get
    $\eforces{w'}{\Psi_i}{\update{e}{\by}{\update{e'}{\bz_i}{\update{e''}{\bx}{\interp{\sigma}_{\update{e}{\by}{\update{e'}{\bz_i}{e''}}}}}}}$,
    and since $\sigma$ is \kl{capture-avoiding} in $\phi$ we have $\fv(\sigma(x))
    \cap \bz_i = \emptyset$ for any $x \in \bx$, and thus
    $\eforces{w'}{\Psi_i}{\update{e}{\by}{\update{e'}{\bz_i}{\update{e''}{\bx}{\interp{\sigma}_{e}}}}}$
    by Observation \ref{obs:update-agree}. Finally by hypothesis $\bx \cap \bz_i
    = \emptyset$, thus we can conclude that
    $\eforces{w'}{\Psi_i}{\update{e}{\bx}{\update{\interp{\sigma}_{e}}{\by}{\update{e'}{\bz_i}{e''}}}}$
    by Observation \ref{obs:update-comm}.
  \end{description}
\end{proof}

The following \emph{functoriality} lemma is at the heart of every deep inference
formalism. It requires an induction principle for \emph{\kl{contexts}}:

\begin{definition}[Depth]
\AP
  The \intro(ctx){depth} $\intro*\cdepth{\cPhi}$ of a \kl{context} $\cPhi$ is
  defined recursively by:
  \begin{mathpar}
    \cdepth{\hole} = 0
    \and
    \cdepth{\Psi, \cphi} = \cdepth{\cphi}
    \and
    \cdepth{\flower{\garden{\bx}{\cPhi}}{\Delta}} =
    \cdepth{\flower{\gamma}{\garden{\bx}{\cPhi} \sep \Delta}} = 1 +
    \cdepth{\cPhi}
\end{mathpar}
\end{definition}

\begin{lemma}[Functoriality]\label{lem:flowers-functoriality}
  
  If $\Phi \kentails{} \Psi$, then for any $\cXi$ either $\cfill{\cXi}{\Phi}
  \kentails{} \cfill{\cXi}{\Psi}$ if $\cXi$ is \kl{positive}, or $\cfill{\cXi}{\Psi}
  \kentails{} \cfill{\cXi}{\Phi}$ if $\cXi$ is \kl{negative}.
\end{lemma}
\begin{proof}
  By induction on $\cdepth{\cXi}$.
\end{proof}

\begin{lemma}[Weakening]\label{lem:flowers-weakening}
  $\Phi \kentails{} \emptyset$.
\end{lemma}
\begin{proof}
  Trivial by Definition \ref{def:forcing}.
\end{proof}

\begin{lemma}[Co-weakening]\label{lem:flowers-coweakening}
  $\flower{\gamma}{\Delta} \kentails{} \flower{\gamma}{\Gamma \sep \Delta}$.
\end{lemma}
\begin{proof}
  Let $\gamma = \garden{\bx}{\Phi}$, $w$ a \kl{world} in some \kl{Kripke structure}
  $\mathcal{K}$, $w' \revaccess w$, $e$ a \evaluation{w} and $e'$ a \evaluation{w'}
  such that $\eforces{w}{\flower{\gamma}{\Delta}}{e}$ and $\eforces{w'}{\Phi}{e
  \upd{\bx} e'}$. Then by hypothesis there must exist some $\garden{\by}{\Psi}
  \in \Delta$ and \evaluation{w'} $e''$ such that $\eforces{w'}{\Psi}{e
  \upd{\bx} e' \upd{\by} e''}$, and thus we can conclude.
\end{proof}

The less obvious rules in terms of soundness are the \emph{pollination} rules
$\{\kl{poll{\da}},\kl{poll{\ua}}\}$, because of the arbitrary \kl{context} $\cXi$
and reliance on the \kl{pollination} relation.

\begin{lemma}[Cross-pollination]\label{lem:flowers-cross-pollination}
  $\Phi, \cfill{\cXi}{\Phi} \kequiv{} \Phi, \cfill{\cXi}{}$.
\end{lemma}
\begin{proof}
  Let $w$ a \kl{world} in some \kl{Kripke structure} $\mathcal{K}$, and $e$ a
  \evaluation{w}. We show that $\eforces{w}{\Phi, \cfill{\cXi}{\Phi}}{e}$ iff
  $\eforces{w}{\Phi, \cfill{\cXi}{}}{e}$ by induction on $\cdepth{\cXi}$.
  \begin{description}
    \item[Base case]
      Suppose $\cXi = \Xi', \hole$. Then we trivially have
      $\eforces{w}{\Phi, \Xi', \Phi}{e}$ iff $\eforces{w}{\Phi, \Xi'}{e}$ by
      Definition \ref{def:forcing}.
    \item[Recursive case]
      We distinguish two cases:
      \begin{description}
        \newcommand{\FillXi}[1]{\Xi', (\flower{\garden{\bx}{#1}}{\Delta})}
        \newcommand{\rFillXi}[1]{\flower{\garden{\bx}{#1}}{\Delta}}
        \newcommand{\fillXi}[1]{\FillXi{\cfill{\cXi_0}{#1}}}
        \newcommand{\rfillXi}[1]{\rFillXi{\cfill{\cXi_0}{#1}}}
        \newcommand{\ffillXi}[1]{\cfill{\cXi}{#1}}
        \item[Pistil]
          Suppose $\cXi = \FillXi{\cXi_0}$.
          \begin{enumerate}
            \item Suppose $\eforces{w}{\Phi, \ffillXi{\Phi}}{e}$. Then
            $\eforces{w}{\Phi}{e}$, $\eforces{w}{\Xi'}{e}$ and
            $\eforces{w}{\rfillXi{\Phi}}{e}$. Thus it remains to show that
            $\eforces{w}{\rfillXi{}}{e}$. Let $w' \revaccess w$ and $e'$ a
            \evaluation{w'} such that
            $\eforces{w'}{\cfill{\cXi_0}{}}{\update{e}{\bx}{e'}}$. By IH we have
            $\Phi, \cfill{\cXi_0}{} \kentails{} \Phi, \cfill{\cXi_0}{\Phi}$, and
            thus by Lemma \ref{lem:flowers-functoriality} $\rFillXi{\Phi,
            \cfill{\cXi_0}{\Phi}} \kentails{} \rFillXi{\Phi, \cfill{\cXi_0}{}}$.
            By Lemma \ref{lem:flowers-weakening} and Lemma
            \ref{lem:flowers-functoriality} we have $\eforces{w}{\rFillXi{\Phi,
            \cfill{\cXi_0}{\Phi}}}{e}$, and thus $\eforces{w}{\rFillXi{\Phi,
            \cfill{\cXi_0}{}}}{e}$. Then since
            $\eforces{w'}{\cfill{\cXi_0}{}}{\update{e}{\bx}{e'}}$, and since by
            Lemma \ref{lem:access-mono} (and the fact that $\bx \cap \fv(\Phi) =
            \emptyset$) we have $\eforces{w'}{\Phi}{\update{e}{\bx}{e'}}$, we
            can conclude that there are some $\garden{\by}{\Psi} \in \Delta$ and
            \evaluation{w'} $e''$ such that
            $\eforces{w'}{\Psi}{\update{e}{\bx}{\update{e'}{\by}{e''}}}$.

            % \item Suppose $w \forces \Phi, (\fillXi{})$. Then $w
            % \forces \Phi$ and $w \forces \fillXi{}$. Thus it
            % remains to show that $w \forces \fillXi{\Phi}$. Let $w' \revaccess w$ and
            % $e$ a \evaluation{w'} such that
            % $\eforces{w'}{\cfill{\cXi_0}{\Phi}}{e}$. By IH we have $\Phi,
            % \cfill{\cXi_0}{\Phi} \kentails{} \Phi, \cfill{\cXi_0}{}$,
            % and thus by \ref{lem:flowers-functoriality} $\FillXi{\Phi,
            % \cfill{\cXi_0}{}} \kentails{} \FillXi{\Phi,
            % \cfill{\cXi_0}{\Phi}}$. By \ref{lem:flowers-weakening} we have $w
            % \forces \FillXi{\Phi, \cfill{\cXi_0}{}}$, and thus $w
            % \forces \FillXi{\Phi, \cfill{\cXi_0}{\Phi}}$. Then since
            % $\eforces{w'}{\cfill{\cXi_0}{\Phi}}{e}$, and since by
            % \ref{lem:access-mono} we have $\eforces{w'}{\Phi}{e}$, we can
            % conclude that $\eforces{w'}{\Delta}{e}$.
            \item $\Phi, \ffillXi{} \kentails{} \Phi, \ffillXi{\Phi}$
            holds by the same argument in the other direction.
          \end{enumerate}

        \renewcommand{\FillXi}[1]{\Xi', (\flower{\garden{\bx}{\Psi}}{\garden{\by}{#1}
        \sep \Delta})}
        \renewcommand{\rFillXi}[1]{\flower{\garden{\bx}{\Psi}}{\garden{\by}{#1}
        \sep \Delta}}
        \item[Petal]
          Suppose $\cXi = \FillXi{\cXi_0}$.
          \begin{enumerate}
            \item Suppose $\eforces{x}{\Phi, \ffillXi{\Phi}}{e}$. Then
            $\eforces{w}{\Phi}{e}$, $\eforces{w}{\Xi'}{e}$ and
            $\eforces{w}{\rfillXi{\Phi}}{e}$. Thus it remains to show that
            $\eforces{w}{\rfillXi{}}{e}$. Let $w' \revaccess w$ and $e'$
            a \evaluation{w'} such that
            $\eforces{w'}{\Psi}{\update{e}{\bx}{e'}}$. Then we can deduce
            that there exists a \evaluation{w'} $e''$ such that either:
            \begin{itemize}
              \item
              $\eforces{w'}{\Psi'}{\update{e}{\bx}{\update{e'}{\mathbf{y'}}{e''}}}$
              for some $\garden{\mathbf{y'}}{\Psi'} \in \Delta$, and we conclude
              immediately;
              \item
              or
              $\eforces{w'}{\cfill{\cXi_0}{\Phi}}{\update{e}{\bx}{\update{e'}{\by}{e''}}}$.
              By Lemma \ref{lem:access-mono} (and the fact that $\bx \cap
              \fv(\Phi) = \emptyset$ and $\by \cap \fv(\Phi) =
              \emptyset$) we have
              $\eforces{w'}{\Phi}{\update{e}{\bx}{\update{e'}{\by}{e''}}}$,
              and thus $\eforces{w'}{\Phi,
              \cfill{\cXi_0}{\Phi}}{\update{e}{\bx}{\update{e'}{\by}{e''}}}$.
              Then by IH we have $\eforces{w'}{\Phi,
              \cfill{\cXi_0}{}}{\update{e}{\bx}{\update{e'}{\by}{e''}}}$,
              and thus we can conclude in particular that $\eforces{w'}{
              \cfill{\cXi_0}{}}{\update{e}{\bx}{\update{e'}{\by}{e''}}}$.
            \end{itemize}

            \item $\Phi, \ffillXi{} \kentails{} \Phi,
            \ffillXi{\Phi}$ holds by the same argument in the other direction.
          \end{enumerate}
      \end{description}
  \end{description}
\end{proof}

\begin{lemma}[Pollination]\label{lem:flowers-pollination}
  
  If $\chyp{\Phi}{\cXi}$, then $\cfill{\cXi}{\Phi} \kequiv{}
  \cfill{\cXi}{}$.
\end{lemma}
\begin{proof}
  We show that $\chyp{\phi}{\cXi}$ implies $\cfill{\cXi}{\phi} \kequiv{}
  \cfill{\cXi}{}$ for any \kl{flower} $\phi$ and \kl{context} $\cXi$: then
  assuming that $\Phi = \phi_1, \ldots, \phi_n$, we get
  $$\underbrace{\cfill{\cXi}{\phi_1, \ldots, \phi_n} \kequiv{} \cfill{\cXi}{\phi_2,
  \ldots, \phi_n} \kequiv{} \ldots \kequiv{} \cfill{\cXi}{}}_{\text{$n$ times}}$$
  and conclude by Observation \ref{obs:kentails-trans}.
  
  By Definition \ref{def:pollination}, there are a \kl{bouquet} $\Psi$ and two
  \kl{contexts} $\cXi', \cXi_0$ such that one of the two following cases holds:
  \begin{description}
    \item[\kl{Cross-pollination}] $\cXi = \cfill{\cXi'}{\Psi, \phi,
    \cXi_0}$. Then $\phi, \cfill{\cXi_0}{\phi} \kequiv{} \phi, \cfill{\cXi_0}{}$
    by Lemma \ref{lem:flowers-cross-pollination}, and we conclude by Lemma
    \ref{lem:flowers-functoriality}.
    \item[\kl{Self-pollination}] $\cXi =
    \cfill{\cXi'}{\flower{\garden{\bx}{\Psi, \phi}}{\garden{\by}{\cXi_0} \sep
    \Delta}}$ for some $\bx,\by,\Delta$. Let $w$ a \kl{world} in some \kl{Kripke
    structure} $\mathcal{K}$ and $e$ a \evaluation{w}. We show that
    $\eforces{w}{\flower{\garden{\bx}{\Psi,
    \phi}}{\garden{\by}{\cfill{\cXi_0}{\phi} \sep \Delta}}}{e}$ iff
    $\eforces{w}{\flower{\garden{\bx}{\Psi, \phi}}{\garden{\by}{\cfill{\cXi_0}{}
    \sep \Delta}}}{e}$, and conclude by Lemma \ref{lem:flowers-functoriality}.
    \begin{enumerate}
      \item Suppose that $\eforces{w}{\flower{\garden{\bx}{\Psi,
      \phi}}{\garden{\by}{\cfill{\cXi_0}{\phi} \sep \Delta}}}{e}$, and let
      $w' \revaccess w$ and $e'$ a \evaluation{w'} such that $\eforces{w'}{\Psi,
      \phi}{\update{e}{\bx}{e'}}$. Then we can deduce that there exists a
      \evaluation{w'} $e''$ such that either:
      \begin{itemize}
        \item
        $\eforces{w'}{\Psi'}{\update{e}{\bx}{\update{e'}{\mathbf{y'}}{e''}}}$
        for some $\garden{\mathbf{y'}}{\Psi'} \in \Delta$, and we conclude
        immediately;
        \item
        or
        $\eforces{w'}{\cfill{\cXi_0}{\phi}}{\update{e}{\bx}{\update{e'}{\by}{e''}}}$.
        Since $\fv(\phi) \cap \by = \emptyset$ we have
        $\eforces{w'}{\phi}{\update{e}{\bx}{\update{e'}{\by}{e''}}}$,
        and thus $\eforces{w'}{\phi,
        \cfill{\cXi_0}{\phi}}{\update{e}{\bx}{\update{e'}{\by}{e''}}}$.
        Then by Lemma \ref{lem:flowers-cross-pollination} we have $\eforces{w'}{\phi,
        \cfill{\cXi_0}{}}{\update{e}{\bx}{\update{e'}{\by}{e''}}}$,
        and thus we can conclude in particular that
        $\eforces{w'}{\cfill{\cXi_0}{}}{\update{e}{\bx}{\update{e'}{\by}{e''}}}$.
      \end{itemize}
    \item $\flower{\garden{\bx}{\Psi,
    \phi}}{\garden{\by}{\cfill{\cXi_0}{} \sep \Delta}}
    \kentails{} \flower{\garden{\bx}{\Psi,
    \phi}}{\garden{\by}{\cfill{\cXi_0}{\phi}} \sep \Delta}$ holds by the same
    argument in the other direction.
    \end{enumerate}
  \end{description}
\end{proof}

Proving the soundness of rules involving \kl{binders} (\kl{ipis}, \kl{ipet},
\kl{apis}, \kl{apet}) is also quite tedious, which can be understood as stemming
from the fact that \kl{substitutions} simulate the complex dynamics of the
\kl{LoIs} of \kl{EGs} in a \emph{global} rather than local way. In particular,
one needs to be careful about the scope of bound variables, which in \kl{EGs}
would be handled locally with (de)iteration rules on \kl{LoIs}.

\begin{lemma}[Universal instantiation]\label{lem:flowers-univ-inst}
  
  If $\sigma : \by$ is \kl{capture-avoiding} in $\flower{\Phi}{\Delta}$,
  then $\flower{\garden{\bx, \by}{\Phi}}{\Delta} \kentails{}
  \flower{\garden{\bx}{\sigma(\Phi)}}{\sigma(\Delta)}$.
\end{lemma}
\begin{proof}
  Let $w$ a \kl{world} in some \kl{Kripke structure} $\mathcal{K}$, $w' \revaccess w$, $e$
  a \evaluation{w} and $e'$ a \evaluation{w'} such that
  $\eforces{w}{\flower{\garden{\bx,\by}{\Phi}}{\Delta}}{e}$ and
  $\eforces{w'}{\sigma(\Phi)}{\update{e}{\bx}{e'}}$. Therefore
  $\eforces{w'}{\Phi}{\update{e}{\bx}{\update{e'}{\by}{\interp{\sigma}_{\update{e}{\bx}{e'}}}}}$
  by Lemma \ref{lem:mirroring}, and thus $\eforces{w'}{\Phi}{\update{e}{\bx \cup
  \by}{(\update{e'}{\by}{\interp{\sigma}_{\update{e}{\bx}{e'}}})}}$ by
  Observation \ref{obs:update-assoc}. Then by hypothesis, there must be some
  $\garden{\bz}{\Psi} \in \Delta$ and \evaluation{w'} $e''$ such that
  $\eforces{w'}{\Psi}{e \upd{\bx \cup \by} (e' \upd{\by} \interp{\sigma}_{e
  \upd{\bx} e'}) \upd{\bz} e''}$, and thus $\eforces{w'}{\Psi}{e \upd{\bx} e'
  \upd{\by} \interp{\sigma}_{e \upd{\bx} e'} \upd{\bz} e''}$. Since $\sigma$ is
  \kl{capture-avoiding} in $\flower{\Phi}{\Delta}$, we know that for any $x \in \by$
  we have $\fv(\sigma(x)) \cap \bz = \emptyset$, and thus $\interp{\sigma(x)}_{e
  \upd{\bx} e' \upd{\bz} e''} = \interp{\sigma(x)}_{e \upd{\bx} e'}$. Hence by
  Observation \ref{obs:update-agree} and Observation \ref{obs:update-comm} we get $\eforces{w'}{\Psi}{e
  \upd{\bx} e' \upd{\bz} e'' \upd{\by} \interp{\sigma}_{e \upd{\bx} e' \upd{\bz}
  e''}}$, and by Lemma \ref{lem:mirroring} we conclude that
  $\eforces{w'}{\sigma(\Psi)}{e \upd{\bx} e' \upd{\bz} e''}$.
\end{proof}

\begin{lemma}[Existential instantiation]\label{lem:flowers-ex-inst}
  
  If $\sigma : \by$ is \kl{capture-avoiding} in $\Phi$, then
  $\flower{\gamma}{\garden{\bx}{\sigma(\Phi)} \sep \Delta} \kentails{}
  \flower{\gamma}{\garden{\bx,\by}{\Phi} \sep \Delta}$.
\end{lemma}
\begin{proof}
  Let $\gamma = \garden{\bz}{\Xi}$, and $w$ a \kl{world} in some \kl{Kripke
  structure} $\mathcal{K}$, $w' \revaccess w$, $e$ a \evaluation{w} and $e'$ a
  \evaluation{w'} such that
  $\eforces{w}{\flower{\gamma}{\garden{\bx}{\sigma(\Phi)} \sep \Delta}}{e}$ and
  $\eforces{w'}{\Xi}{\update{e}{\bz}{e'}}$. Then by hypothesis, there must be
  some \evaluation{w'} $e''$ such that either:
  \begin{itemize}
    \item
    $\eforces{w'}{\Xi'}{\update{e}{\bz}{\update{e'}{\bz'}{e''}}}$ for some
    $\garden{\bz'}{\Xi'} \in \Delta$, and we conclude immediately;
    \item
    or $\eforces{w'}{\sigma(\Phi)}{\update{e}{\bz}{\update{e'}{\bx}{e''}}}$.
    Then by Lemma \ref{lem:mirroring} we have
    $\eforces{w'}{\Phi}{\update{e}{\bz}{\update{e'}{\bx}{e'' \upd{\by}
    \interp{\sigma}_{\update{e}{\bz}{\update{e'}{\bx}{e''}}}}}}$, and thus we
    can conclude with $\eforces{w'}{\Phi}{\update{e}{\bz}{\update{e'}{\bx \cup
    \by}{(e'' \upd{\by}
    \interp{\sigma}_{\update{e}{\bz}{\update{e'}{\bx}{e''}}})}}}$ by
    Observation \ref{obs:update-agree}.
  \end{itemize}
\end{proof}

We are now equipped with enough lemmas to prove the soundness of each rule,
starting with the \emph{\kl{shallow}} version of \kl{natural} rules. In fact we are
able to prove more: that every $\Nature$-rule is \emph{invertible}, i.e. its
conclusion entails its premiss.

\begin{lemma}[Shallow soundness]\label{lem:flowers-local-soundness} If $\Phi
  \lstep \Psi$, then $\Phi \kequiv{} \Psi$.
\end{lemma}
\begin{proof}
  Let $w$ a \kl{world} in some \kl{Kripke structure} $\mathcal{K}$, $w' \revaccess w$, $e$
  a \evaluation{w} and $e'$ a \evaluation{w'}. We proceed by inspection of every
  $\Nature$-rule.
  
  \begin{description}
    \item[\kl{poll{\da}}, \kl{poll{\ua}}]
      By Lemma \ref{lem:flowers-pollination}.
    
    \item[\kl{epis}]~\\\vspace{-1.5em}
    \begin{enumerate}
      \item Suppose that $\eforces{w}{\Phi}{e}$. Then by Lemma \ref{lem:access-mono}
      we have $\eforces{w'}{\Phi}{e}$, and thus we can conclude for instance with
    $\eforces{w'}{\Phi}{\update{e}{\emptyset}{\update{e'}{\emptyset}{e}}}$.
      \item Suppose that $\eforces{w}{\flower{{}}{\Phi}}{e}$. Then since we
      trivially have $w \revaccess w$ and
      $\eforces{w}{\emptyset}{\update{e}{\emptyset}{e}}$, we get that
      $\eforces{w}{\Phi}{\update{e}{\emptyset}{\update{e}{\emptyset}{e''}}}$
      for some \evaluation{w} $e''$, and thus $\eforces{w}{\Phi}{e}$.
    \end{enumerate}

    \item[\kl{epet}]
      Let $\gamma = \garden{\bx}{\Phi}$. We trivially have that
      $\eforces{w'}{\emptyset}{\update{e}{\bx}{\update{e'}{\emptyset}{e}}}$,
      and thus can conclude.

    \item[\kl{ipis}] We trivially have $\flower{\garden{\bx,\by}{\Phi}}{\Delta}
      \kentails{} \flower{\garden{\bx,\by}{\Phi}}{\Delta}$ by
      Observation \ref{obs:kentails-refl}, and thus we can conclude by
      Lemma \ref{lem:flowers-univ-inst}.

    \item[\kl{ipet}] The first direction is trivial by
      Lemma \ref{lem:flowers-coweakening}. In the other direction, let $\gamma =
      \garden{\bz}{\Xi}$, and suppose that
      $\eforces{w}{\flower{\gamma}{\garden{\bx}{\sigma(\Phi)} \sep
      \garden{\bx,\by}{\Phi} \sep \Delta}}{e}$ and $\eforces{w'}{\Xi}{e
      \upd{\bz} e'}$. Then there must be some \evaluation{w'} $e''$ such that
      either:
      \begin{itemize}
        \item $\eforces{w'}{\Xi'}{\update{e}{\bz}{\update{e'}{\bz'}{e''}}}$ for
        some $\garden{\bz'}{\Xi'} \in \Delta$, and we conclude immediately;
        \item $\eforces{w'}{\Phi}{e \upd{\bz} e' \upd{\bx \cup \by} e''}$, and
        we also conclude immediately;
        \item or $\eforces{w'}{\sigma(\Phi)}{e \upd{\bz} e' \upd{\bx} e''}$, and
        we conclude with the same argument as in the proof of
        Lemma \ref{lem:flowers-ex-inst}.
      \end{itemize}
      
    \item[\kl{srep}]
      Let $\gamma_i = \garden{\by_i}{\Psi_i}$ for $1 \leq i \leq n$.
      \begin{enumerate}
        \item Suppose that
        $\eforces{w}{\flower{\garden{\bx}{\Phi,(\flower{}{\fset{i}{n}{\gamma_i}})}}{\Delta}}{e}$
        and $\eforces{w'}{\Phi}{e \upd{\bx} e'}$. We show that
        $\eforces{w'}{\flower{\gamma_i}{\Delta}}{e \upd{\bx} e'}$ for all $1
        \leq i \leq n$, i.e. for every $w'' \revaccess w'$ and \evaluation{w''} $e''$,
        $\eforces{w''}{\Psi_i}{e \upd{\bx} e' \upd{\by_i} e''}$ implies that
        there is some $\garden{\bz}{\Xi} \in \Delta$ and \evaluation{w''} $e'''$
        such that $\eforces{w''}{\Xi}{e \upd{\bx} e' \upd{\by_i} e'' \upd{\bz}
        e'''}$. By assumption, Lemma \ref{lem:access-mono} and the fact that
        $\fv(\Phi) \cap \by_i = \emptyset$, we have $\eforces{w''}{\Phi}{e
        \upd{x} e' \upd{\by_i} e''}$. Also since $\eforces{w''}{\Psi_i}{e
        \upd{\bx} e' \upd{\by_i} e''}$ we immediately get
        $\eforces{w''}{\flower{}{\fset{i}{n}{\gamma_i}}}{e \upd{\bx} e'}$, and
        thus $\eforces{w''}{\flower{}{\fset{i}{n}{\gamma_i}}}{e \upd{\bx} e'
        \upd{\by_i} e''}$ since $\fv(\flower{}{\fset{i}{n}{\gamma_i}}) \cap
        \by_i = \emptyset$. Thus by Observation \ref{obs:update-assoc} we have
        $\eforces{w''}{\Phi,(\flower{}{\fset{i}{n}{\gamma_i}})}{e \upd{\bx \cup
        \by_i} (e' \upd{\by_i} e'')}$, and by hypothesis (and the fact that $w''
        \revaccess w$ by transitivity) we obtain that $\eforces{w''}{\Xi}{e \upd{\bx
        \cup \by_i} (e' \upd{\by_i} e'') \upd{\bz} e'''}$ for some
        $\garden{\bz}{\Xi} \in \Delta$ and \evaluation{w''} $e'''$. Then we
        conclude again by Observation \ref{obs:update-assoc}.

        \item Suppose that
        $\eforces{w}{\flower{\garden{\bx}{\Phi}}{{\fset{i}{n}{\flower{\gamma_i}{\Delta}}}}}{e}$
        and $\eforces{w'}{\Phi,(\flower{}{\fset{i}{n}{\gamma_i}})}{e \upd{\bx}
        e'}$. Then there must be some $1 \leq i \leq n$ and \evaluation{w'}
        $e''$ such that $\eforces{w'}{\Psi_i}{e \upd{\bx} e' \upd{\by_i} e''}$,
        and for all $1 \leq j \leq n$ we know that
        $\eforces{w'}{\flower{\gamma_j}{\Delta}}{e \upd{\bx} e'}$. Thus since
        $w' \leq w'$ by reflexivity, there must be some $\garden{\bz}{\Xi} \in
        \Delta$ and \evaluation{w'} $e'''$ such that $\eforces{w'}{\Xi}{e
        \upd{\bx} e' \upd{\by_i} e'' \upd{\bz} e'''}$, and we can conclude with
        $\eforces{w'}{\Xi}{e \upd{\bx} e' \upd{\by_i \cup \bz} (e'' \upd{\bz}
        e''')}$ by Observation \ref{obs:update-assoc}.
      \end{enumerate} 
  \end{description}
\end{proof}

Then the soundness of the contextual closure of \kl{natural} rules follows
immediately from functoriality:

\begin{lemma}[Natural soundness]\label{lem:flowers-natural-soundness}
  If $\Phi \step{\Nature} \Psi$ then $\Phi \kequiv{} \Psi$.
\end{lemma}
\begin{proof}
  By Lemma \ref{lem:flowers-local-soundness} and Lemma \ref{lem:flowers-functoriality}.
\end{proof}

The soundness of \kl{cultural} rules is straightforward with the previous lemmas:

\begin{lemma}[Cultural soundness]\label{lem:flowers-cultural-soundness}
  If $\Phi \step{\Culture} \Psi$ then $\Psi \kentails{} \Phi$.
\end{lemma}
\begin{proof}
  By inspection of every $\Culture$-rule.
  \begin{description}
    \item[\kl{grow}, \kl{crop}] By Lemma \ref{lem:flowers-weakening} and
    Lemma \ref{lem:flowers-functoriality}.
    
    \item[\kl{pull}, \kl{glue}] By Lemma \ref{lem:flowers-coweakening} and
    Lemma \ref{lem:flowers-functoriality}.
    
    \item[\kl{apis}] By Lemma \ref{lem:flowers-univ-inst} and Lemma \ref{lem:flowers-functoriality}.

    \item[\kl{apet}] By Lemma \ref{lem:flowers-ex-inst} and Lemma \ref{lem:flowers-functoriality}.
  \end{description}
\end{proof}

Then it follows that every \kl{derivation} in the \kl{flower calculus} is sound:

\begin{theorem}\label{thm:flowers-soundness}
  If $\Phi \steps{} \Psi$ then $\Psi \kentails{} \Phi$.
\end{theorem}
\begin{proof}
  By Lemma \ref{lem:flowers-natural-soundness}, Lemma \ref{lem:flowers-cultural-soundness}
  and Observation \ref{obs:kentails-refl}, Observation \ref{obs:kentails-trans}.
\end{proof}

In particular $\entails{} \phi$ implies $\kentails{} \phi$, i.e. every provable
\kl{flower} is true.

\section{Detailed proofs}\label{app:Proofs}

\subsection{Deduction theorems}\label{app:Proofs-deduction}

\begin{lemma}[Positive closure]\label{lem:flowers-posclos}
  If $\Phi \step{} \Psi$, then $\cfill{\cXiP}{\Phi} \step{} \cfill{\cXiP}{\Psi}$.
\end{lemma}
\begin{proof}
  In the case of a \kl{natural} step $\Phi \step{\Nature} \Psi$, this is immediate
  by contextual closure (Definition \ref{def:derivation}). Otherwise we have a
  \kl{cultural} step $\cfill{\cXi'}{\Phi_0} \step{\Culture} \cfill{\cXi'}{\Psi_0}$.
  Then either $\cXi'$ is \kl{positive}, and $\inv(\cfill{\cXiP}{\cXi'}) = \inv(\cXiP)
  + \inv(\cXi')$ is even since it is the sum of two even numbers; or $\cXi'$ is
  \kl{negative}, and $\inv(\cfill{\cXiP}{\cXi'})$ is odd since it is the sum of an
  even and an odd number. In both cases $\cfill{\cXiP}{\cXi'}$ has the same
  polarity as $\cXi'$, and thus the same rule can be applied.
\end{proof}

\subsubsection{Proof of Theorem \ref{thm:flowers-strong-deduction}}

\begin{proof}
  Suppose that $\Phi \steps{} \Psi$. Then we have:
  \[
  \begin{array}{rlll}
    \flower{\Psi}{\Phi}
    &\steps{} &\flower{\Psi}{\Psi} &\text{(Hypothesis + Lemma \ref{lem:flowers-posclos})}\\
    &\step{\kl{poll{\da}}} &\flower{\Psi}{\cdot} &\\
    &\step{\kl{epet}} &\emptyset &
  \end{array}
  \]
  
  In the other direction, suppose that $\flower{\Psi}{\Phi} \steps{} \emptyset$.
  Then we have:
  \[
  \begin{array}{rlll}
    \Phi
    &\step{\kl{epis}} &\flower{}{\Phi} &\\
    &\step{\kl{grow}} &(\flower{\Psi}{\Phi}),(\flower{}{\Phi}) &\\
    &\step{\kl{poll{\ua}}} &(\flower{\Psi}{\Phi}),(\flower{(\flower{\Psi}{\Phi})}{\Phi}) &\\
    &\steps{} &\flower{(\flower{\Psi}{\Phi})}{\Phi} &\text{(Hypothesis + Lemma \ref{lem:flowers-posclos})}\\
    &\step{\kl{grow}} &\Psi,(\flower{(\flower{\Psi}{\Phi})}{\Phi}) &\\
    &\step{\kl{poll{\da}}} &\Psi,(\flower{(\flower{}{\Phi})}{\Phi}) &\\
    &\step{\kl{srep}} &\Psi,(\flower{}{(\flower{\Phi}{\Phi})}) &\\
    &\step{\kl{poll{\da}}} &\Psi,(\flower{}{(\flower{\Phi}{\cdot})}) &\\
    &\step{\kl{epet}} &\Psi,(\flower{}{\cdot}) &\\
    &\step{\kl{epet}} &\Psi
  \end{array}
  \]
\end{proof}

\subsubsection{Proof of Theorem \ref{thm:flowers-deduction}}

\begin{proof}
  Let $\cXi$ be some \kl{context}. If $\Psi \entails{\Nature} \Phi$, then in
  particular $\cfill{\cXi'}{\Phi} \steps{\Nature} \cfill{\cXi'}{}$ for $\cXi'
  \deq \cfill{\cXi}{\flower{\Psi}{\hole}}$. Thus we have:
  \[
  \cfill{\cXi}{\flower{\Psi}{\Phi}} \steps{\Nature}
  \cfill{\cXi}{\flower{\Psi}{\garden{{}}{{}}}} \step{\kl{epet}}
  \cfill{\cXi}{}
  \]
  In the other direction, let $\cXi$ be some \kl{context} such that
  $\chyp{\Psi}{\cXi}$. If $\entails{\Nature}
  \flower{\Psi}{\Phi}$, then in particular
  $\cfill{\cXi}{\flower{\Psi}{\Phi}} \steps{\Nature}
  \cfill{\cXi}{}$. Thus we have:
  \[
  \cfill{\cXi}{\Phi} \step{\kl{epis}}
  \cfill{\cXi}{\flower{{}}{\Phi}} \step{\kl{poll{\ua}}}
  \cfill{\cXi}{\flower{\Psi}{\Phi}} \steps{\Nature}
  \cfill{\cXi}{}
  \]
\end{proof}

\subsection{Completeness}\label{app:Proofs-completeness}

\begin{lemma}[Reflexivity]\label{lem:reflexivity}
  For any \kl{bouquet} $\Phi$, $\Phi \entails{\Nature} \Phi$.
\end{lemma}
\begin{proof}
  Trivial by application of the \kl{poll{\da}} rule.
\end{proof}

\begin{lemma}[Weakening]\label{lem:weakening}
  If $\tT \subseteq \tT'$ and $\tT \entails{} \phi$, then $\tT' \entails{} \phi$.
\end{lemma}
\begin{proof}
  This follows immediately from Definition \ref{def:theory}.
\end{proof}

In the following, we suppose some enumeration $(\phi_n)_{n \in \nats}$ of
$\flowers$, which should be definable constructively given the inductive nature
of \kl{flowers}. Let $\psi \in \flowers$, and $\tT$ a \consistent{\psi} \kl{theory}. We
now define the \emph{completion procedure}, which constructs an extension
$\completion{\tT} \supseteq \tT$ with the property that $\completion{\tT}$ is
\consistent{\psi} and \complete{\psi}.

% \subsubsection{Proof of Lemma \ref{lem:completion-consistent-complete}}

\begin{definition}[$n$-completion]
\AP
  The \emph{$n$-completion} $\intro*\ncompletion{\tT}{n}$ of $\tT$ is defined
  recursively by:
  \begin{mathpar}
    \ncompletion{\tT}{0} = \tT \and
    \ncompletion{\tT}{n+1} =
    \begin{cases}
      \ncompletion{\tT}{n} \cup \phi_n &\text{if $\ncompletion{\tT}{n} \cup \phi_n$ is \consistent{\psi}} \\
      \ncompletion{\tT}{n} &\text{otherwise}
    \end{cases}
  \end{mathpar}
\end{definition}

\begin{definition}[Completion]
\AP
  The \emph{completion} $\intro*\completion{\tT}$ of $\tT$ is the denumerable
  union of all $n$-completions:
  $$\completion{\tT} = \bigcup_{n \in \nats}{\ncompletion{\tT}{n}}$$
\end{definition}

\begin{lemma}\label{lem:completion-consistent-complete}
  For every $\psi \in \flowers$, $\completion{\tT}$ is \consistent{\psi} and
  \complete{\psi}.
\end{lemma}
\begin{proof}
  For \consistency{\psi}, it is immediate by induction on $n$ that
  $\ncompletion{\tT}{n}$ is \consistent{\psi}. Then suppose that
  $\completion{\tT} \entails{\Nature} \psi$, that is there is some \kl{bouquet}
  $\Phi \subseteq \completion{\tT}$ such that $\Phi \entails{\Nature} \psi$. For
  each $\phi \in \Phi$, there is some rank $n$ such that $\phi \in
  \ncompletion{\tT}{n}$. Let $m$ be the greatest such rank. Then $\Phi \subseteq
  \ncompletion{\tT}{m}$, and thus by weakening (Lemma \ref{lem:weakening}) $\Phi
  \nentails{\Nature} \psi$. Contradiction.
  
  For \completeness{\psi}, suppose that there is some $\phi$ such that
  $\completion{\tT}, \phi \nentails{\Nature} \psi$ and $\phi \not\in
  \completion{\tT}$, and let $\phi = \phi_n$. Then $\completion{\tT} \cup
  \phi_n$ is \consistent{\psi}, and thus by weakening (Lemma
  \ref{lem:weakening}) so is $\ncompletion{\tT}{n} \cup \phi_n$. This entails
  that $\phi_n \in \ncompletion{\tT}{n+1} \subseteq \completion{\tT}$.
  Contradiction.
\end{proof}

\subsubsection{Proof of Lemma \ref{lem:adequacy}}

\begin{proof}
  By induction on $\fdepth{\phi}$.
  \begin{itemize}
    \item Suppose $\phi = p(\tvec{x})$.
    \begin{enumerate}
      \item By definition of \kl{forcing} (Definition \ref{def:forcing}) and $\rewolF{\psi}$
      (Definition \ref{def:kcanon}), $\eforces{\tT}{p(\tvec{x})}{\sigma}$ precisely when
      $\sigma(p(\tvec{x})) \in \tT$.
      \item Suppose that $\eforces{\tT}{\phi}{\sigma}$, that is $\sigma(\phi)
      \in \tT$. Then by weakening (Lemma \ref{lem:weakening}), we get $\sigma(\phi)
      \nentails{\Nature} \sigma(\phi)$. But this is impossible by reflexivity of
      ${\entails{}}$ (Lemma \ref{lem:reflexivity}).
    \end{enumerate}
    \item Suppose $\phi =
    \flower{\garden{\bx}{\Phi}}{\fset{i}{n}{\garden{\bx_i}{\Phi_i}}}$.
    \begin{enumerate}
      \item Let $\tU \supseteq \tT$ be a \consistent{\psi} and \complete{\psi}
      \kl{theory}. Obviously $\sigma(\phi) =
      \flower{\garden{\bx}{\restr{\sigma}{\bx}(\Phi)}}{\fset{i}{n}{\garden{\bx_i}{\restr{\sigma}{\bx
      \cup \bx_i}(\Phi_i)}}} \in \tU$, and thus by Proposition \ref{prop:inv-elem-flower},
      for every \kl{substitution} $\tau$, either $ \update{}{\bx}{\tau} \circ
      \restr{\sigma}{\bx \cup \bx_i}(\Phi_i) =
      \update{\update{\sigma}{\bx}{\tau}}{\bx_i}{}(\Phi_i) \subseteq \tU$ for
      some $1 \leq i \leq n$, or $\tU \nentails{\Nature} \update{}{\bx}{\tau} \circ
      \restr{\sigma}{\bx}(\Phi) = \update{\sigma}{\bx}{\tau}(\Phi)$. In the
      first case, we get
      $\eforces{\tU}{\Phi_i}{\update{\update{\sigma}{\bx}{\tau}}{\bx_i}{}}$ by
      IH. In the second case, we get
      $\neforces{\tU}{\Phi}{\restr{\sigma}{\bx}\tau}$ by IH. In other words,
      $\eforces{\tU}{\Phi}{\restr{\sigma}{\bx}\tau}$ implies
      $\eforces{\tU}{\Phi_i}{\update{\update{\sigma}{\bx}{\tau}}{\bx_i}{}}$,
      that is $\eforces{\tT}{\phi}{\sigma}$.
      \item By Proposition \ref{prop:inv-deriv-flower}, for every $1 \leq i \leq
      n$ and \kl{substitution} $\tau$, there is some $\phi_i \in \Phi_i$ such
      that $\tT, \restr{\sigma}{\bx}(\Phi) \nentails{\Nature}
      \update{}{\bx_i}{\tau} \circ \restr{\sigma}{\bx \cup \bx_i}(\phi_i) =
      \update{\update{\sigma}{\bx}{}}{\bx_i}{\tau}(\phi_i)$. By the completion
      procedure, we get a \kl{theory} $\tU = \completion{\tT \cup
      \restr{\sigma}{\bx}(\Phi)} \supseteq \tT \cup \restr{\sigma}{\bx}(\Phi)$
      which is both
      \consistent{\update{\update{\sigma}{\bx}{}}{\bx_i}{\tau}(\phi_i)} and
      \complete{\update{\update{\sigma}{\bx}{}}{\bx_i}{\tau}(\phi_i)} (Lemma
      \ref{lem:completion-consistent-complete}). Then by IH,
      $\eforces{\tU}{\Phi}{\update{\sigma}{\bx}{}}$ since
      $\update{\sigma}{\bx}{}(\Phi) \subseteq \tU$, and
      $\neforces{\tU}{\phi_i}{\update{\update{\sigma}{\bx}{}}{\bx_i}{\tau}}$
      since $\tU$ is
      \consistent{\update{\update{\sigma}{\bx}{}}{\bx_i}{\tau}(\phi_i)}, that is
      $\neforces{\tT}{\phi}{\sigma}$.
    \end{enumerate}
  \end{itemize}
\end{proof}

\subsubsection{Proof of Theorem \ref{thm:flowers-completeness}}

% \begin{lemma}\label{lem:completeness-contra}
%   Let $\psi \in \flowers$ and $\tT$ be a \consistent{\psi} theory. Then $\tT
%   \nvDash \psi$.
% \end{lemma}
\begin{proof}
  Let $\tT$ be a \consistent{\psi} \kl{theory}. We prove that $\tT \notkentails{} \psi$ by
  showing in particular that $\tT \notkentails{\rewolF{\psi}} \psi$, and more
  specifically that $\eforces{\completion{\tT}}{\tT}{\idsubst}$ but
  $\neforces{\completion{\tT}}{\psi}{\idsubst}$. Then it follows by (classical)
  contraposition that $\tT \kentails{} \psi$ implies $\tT \entails{\Nature} \psi$ for
  any $\psi$ and any $\tT$, and thus we can conclude.
  \begin{itemize}
    \item Let $\phi \in \tT$. Then $\idsubst(\phi) = \phi \in \completion{\tT}$,
    thus by \kl[\consistent]{$\psi$-consistency} and \kl[\complete]{$\psi$-completeness} of the
    completion (Lemma \ref{lem:completion-consistent-complete}), one can apply
    adequacy (Lemma \ref{lem:adequacy}) to get
    $\eforces{\completion{\tT}}{\phi}{\idsubst}$.
    \item Similarly, we can apply adequacy (Lemma \ref{lem:adequacy}) to get
    $\neforces{\completion{\tT}}{\psi}{\idsubst}$.
  \end{itemize}
\end{proof}

\end{document}

%% file: preamble.tex
\usepackage{pifont}
\usepackage{mathpartir}
\usepackage{tikz}
\usepackage{tikzit}
\usepackage{framed}
\usepackage{prftree}
\usepackage[inline]{enumitem}
\usepackage[notion,electronic]{knowledge}
\input{notions.tex}

\graphicspath{{./}} % Paths where images are looked for

\usetikzlibrary{calc}
\usetikzlibrary{decorations.pathmorphing,decorations.pathreplacing,decorations.shapes}
\usetikzlibrary{arrows.meta}
\input{./eg.tikzstyles.tex}
\newcommand{\tikzscale}[3]{
  \scalebox{#1}{
    \begingroup
      \tikzset{every picture/.style={scale=#2,baseline=(current bounding box.center)}}
      #3
    \endgroup}}
\renewcommand{\tikzfig}[3]{
  \tikzscale{#1}{#2}{\InputIfFileExists{./#3.tex}{}{{\color{red}\colorbox{pink}{missingfile : #3}}}}}
\newcommand{\stkfig}[2]{\input{#1.tikz}{0.5}{#2}}

\makeatletter
\DeclareFontEncoding{LS1}{}{\noaccents@}
\makeatother
\makeatletter
\DeclareFontEncoding{LS2}{}{\noaccents@}
\makeatother
\DeclareFontSubstitution{LS1}{stix}{m}{n}
\DeclareFontSubstitution{LS2}{stix}{m}{n}

\DeclareSymbolFont{xsymbols}{LS1}{stixscr}{m}{n}
\DeclareSymbolFont{xsymbols3}{LS1}{stixbb}{m}{n}
\DeclareSymbolFont{xsymbols4}{LS1}{stixbb}{m}{it}
\DeclareSymbolFont{xlargesymbols}{LS2}{stixex}{m}{n}
\DeclareSymbolFont{xarrows3}{LS2}{stixtt}{m}{n}
\DeclareSymbolFont{xintegrals}{LS2}{stixcal}{m}{n}

\DeclareMathSymbol{\csup}{\mathrel}{xsymbols4}{"9E}
\DeclareMathSymbol{\Coloneq}{\mathrel}{xsymbols3}{"C4}
\DeclareMathSymbol{\mdwhtcircle}{\mathord}{xarrows3}{"7D}
\DeclareMathSymbol{\smalltriangleright}{\mathord}{xsymbols}{"D0}
\DeclareMathSymbol{\coloneq}{\mathrel}{xintegrals}{"85}
\DeclareMathDelimiter{\lBrack}{\mathopen}{xlargesymbols}{"E0}{xlargesymbols}{"06}
\DeclareMathDelimiter{\rBrack}{\mathclose}{xlargesymbols}{"E1}{xlargesymbols}{"07}

\DeclareMathAlphabet{\mathcall}{OMS}{cmsy}{m}{n}

\definecolor{pos}{HTML}{C80000}
\definecolor{neg}{HTML}{0000C8}
\definecolor{hypo}{HTML}{007BFF}
\definecolor{conc}{HTML}{DC3545}
\definecolor{pbg}{HTML}{FFFFFF}
\definecolor{nbg}{HTML}{DFDFDE}
\definecolor{pfg}{HTML}{000000}
\definecolor{nfg}{HTML}{000000}
\definecolor{pistil}{HTML}{FFF344}
\definecolor{hyp}{HTML}{FFEBFF}
\definecolor{hypnum}{HTML}{980098}

\input{notations.tex}

%% file: notions.tex
\definecolor{Dark Ruby Red}{HTML}{7c1b1e}
\definecolor{Dark Blue Sapphire}{HTML}{004c5c}
\definecolor{Dark Gamboge}{HTML}{be7c00}

\IfKnowledgePaperModeTF{
}{
    \knowledgestyle{intro notion}{color={Dark Ruby Red}, emphasize}
    \knowledgestyle{notion}{color={Dark Blue Sapphire}}
    \hypersetup{
        colorlinks=true,
        breaklinks=true,
        linkcolor={Dark Blue Sapphire}, % Links to sections, pages, etc.
        citecolor={Dark Blue Sapphire}, % Links to bibliography
        filecolor={Dark Blue Sapphire}, % Links to local file
        urlcolor={Dark Blue Sapphire},
    }
    \IfKnowledgeElectronicModeTF{
    }{
        \knowledgeconfigure{anchor point color={Dark Ruby Red}, anchor point shape=corner}
        \knowledgestyle{intro unknown}{color={Dark Gamboge}, emphasize}
        \knowledgestyle{intro unknown cont}{color={Dark Gamboge}, emphasize}
        \knowledgestyle{kl unknown}{color={Dark Gamboge}}
        \knowledgestyle{kl unknown cont}{color={Dark Gamboge}}
    }
}

\knowledgedirective{system}{notion,wrap=\sys}
\knowledgedirective{acronym}{synonym,up}
\knowledgedirective{rule}{notion,wrap=\rsf}

\knowledge{notion}
 | interactive theorem provers
\knowledge{acronym}
 | ITPs
 | ITP
 
\knowledge{notion}
 | goal
 | goals

\knowledge{notion}
 | proof state

\knowledge{notion}
 | tactic
 | tactics

\knowledge{notion}
 | graphical user interfaces
\knowledge{acronym}
 | GUIs
 | GUI

\knowledge{notion}
 | subformula linking
\knowledge{acronym}
 | SFL

\knowledge{notion}
 | first-order logic
\knowledge{acronym}
 | FOL

\knowledge{notion}
 | existential graphs
\knowledge{acronym}
 | EGs

\knowledge{notion}
 | flower calculus

\knowledge{notion}
 | Flower Prover

\knowledge{system}
 | Alpha

\knowledge{system}
 | Beta

\knowledge{system}
 | Gamma

\knowledge{notion}
 | sheet of assertion
 
\knowledge{notion}
 | vacuous truth

\knowledge{notion}
 | conjunction
 | conjunctions

\knowledge{notion}
 | negation
 | negations

\knowledge{notion}
 | graph
 | graphs
 | Graph
 | Graphs
 | subgraph
 | subgraphs

\knowledge{notion}
 | falsehood
 | absurdity

\knowledge{notion}
 | disjunction
 | disjunctions

\knowledge{notion}
 | implication
 | implications

\knowledge{notion}
 | juxtaposition

\knowledge{notion}
 | cut
 | cuts

\knowledge{notion}
 | isotropy

\knowledge{notion}
 | illative transformation
 | illative transformations

\knowledge{rule}
 | Iteration

\knowledge{rule}
 | Insertion

\knowledge{rule}
 | Double{-}cut

\knowledge{notion}
 | scroll
 | scrolls

\knowledge{notion}
 | curl
 | curls

\knowledge{notion}
 | inloop
 | inloops

\knowledge{notion}
 | outloop
 | outloops

\knowledge{notion}
 | flower@informal
 | flowers@informal
 | Flower@informal
 | Flowers@informal

\knowledge{notion}
 | flower
 | flowers
 | Flower
 | Flowers

\knowledge{notion}
 | lines of identity
 | line of identity
\knowledge{acronym}
 | LoIs
 | LoI

\knowledge{notion}
 | binder
 | binders
 | Binder
 | Binders

\knowledge{notion}
 | sprinkler
 | sprinklers
 | Sprinkler
 | Sprinklers

\knowledge{notion}
 | garden
 | gardens
 | Garden
 | Gardens

\knowledge{notion}
 | first-order signature

\knowledge{notion}
 | predicate symbols

\knowledge{notion}
 | arity

\knowledge{notion}
 | bouquet
 | bouquets
 | Bouquet
 | Bouquets

\knowledge{notion}
 | corolla
 | corollas
 | Corolla
 | Corollas

\knowledge{notion}
 | pistil
 | pistils
 | Pistil
 | Pistils

\knowledge{notion}
 | petal
 | petals
 | Petal
 | Petals

\knowledge{notion}
 | free variable
 | free variables
 | Free variable
 | Free variables

\knowledge{notion}
 | bound variable
 | bound variables
 | Bound variable
 | Bound variables

\knowledge{notion}
 | Barendregt convention

\knowledge{notion}
 | update

\knowledge{notion}
 | substitution
 | substitutions
 | Substitution
 | Substitutions

\knowledge{notion}
 | support

\knowledge{notion}
 | capture-avoiding

\knowledge{notion}
 | context
 | contexts
 | Context
 | Contexts

\knowledge{notion}
 | hole
 | holes
 | Hole
 | Holes

\knowledge{notion}
 | filling

\knowledge{notion}
 | positive
 | Positive
 | positive@beta
 | Positive@beta

\knowledge{notion}
 | negative
 | Negative
 | negative@beta
 | Negative@beta

\knowledge{notion}
 | pollination
 | Pollination
 | pollinated
 | pollinate
 | pollinates
 | pollinating

\knowledge{notion}
 | cross-pollination
 | Cross-pollination

\knowledge{notion}
 | self-pollination
 | Self-pollination

\knowledge{rule}
 | poll{\da}
\knowledge{rule}
 | poll{\ua}
\knowledge{rule}
 | epis
\knowledge{rule}
 | epet
\knowledge{rule}
 | srep
\knowledge{rule}
 | ipis
\knowledge{rule}
 | ipet
\knowledge{rule}
 | grow
\knowledge{rule}
 | crop
\knowledge{rule}
 | pull
\knowledge{rule}
 | glue
\knowledge{rule}
 | apis
\knowledge{rule}
 | apet

\knowledge{notion}
 | shallow
 | Shallow

\knowledge{notion}
 | natural
 | Natural

\knowledge{notion}
 | cultural
 | Cultural

\knowledge{notion}
 | analytic
 | Analytic
 | analyticity

\knowledge{notion}
 | derivation
 | derivations
 | Derivation
 | Derivations

\knowledge{notion}
 | proof
 | proofs
 | Proof
 | Proofs
 
\knowledge{notion}
 | first-order
 | First-order

\knowledge{notion}
 | higher-order
 | Higher-order

\knowledge{notion}
 | hypothetically provable
 | hypothetical provability
 | provable from

\knowledge{notion}
 | first-order structure
 | first-order structures

\knowledge{notion}
 | Kripke structure
 | Kripke structures

\knowledge{notion}
 | universal Kripke structure
 | universal Kripke structures

\knowledge{notion}
 | domain
 | domains

\knowledge{notion}
 | interpretation

\knowledge{notion}
 | world
 | worlds

\knowledge{notion}
 | accessibility

\knowledge{notion}
 | depth
 | depth@ctx

\knowledge{notion}
 | forcing
 | forced

\knowledge{notion}
 | consequence
 | semantic entailment
 | semantically entails

\knowledge{notion}
 | equivalence
 | equivalent
 | semantic equivalence
 | semantically equivalent

\knowledge{notion}
 | valid
 | validity

\knowledge{notion}
 | theory
 | theories
 | Theory
 | Theories

\knowledge{notion}
 | coherent logic

\knowledge{notion}
 | topos theory
 
\knowledge{notion}
 | $\beta$-node
 | $\beta$-nodes

\knowledge{notion}
 | $\beta$-garden
 | $\beta$-gardens

\knowledge{notion}
 | $\beta$-graph
 | $\beta$-graphs

\knowledge{notion}
 | $\beta$-context
 | $\beta$-contexts
 
\knowledge{notion}
 | nested sequent
 | nested sequents
 | Nested sequent
 | Nested sequents

\knowledge{notion}
 | graph calculus
 | graph calculi
 | Graph calculus
 | Graph calculi

\knowledge{rule}
 | Iter
\knowledge{rule}
 | Deit
\knowledge{rule}
 | Ins
\knowledge{rule}
 | Del
\knowledge{rule}
 | Dcut{\da}
\knowledge{rule}
 | Dcut{\ua}
\knowledge{rule}
 | Unif{\da}
\knowledge{rule}
 | Unif{\ua}
\knowledge{rule}
 | epis{\da}

%% file: eg.tikzstyles.tex
\tikzstyle{neg text}=[fill=none, draw=none, text=nfg, tikzit fill={rgb,255: red,191; green,191; blue,191}, tikzit draw={rgb,255: red,191; green,191; blue,191}]
\tikzstyle{source}=[fill=hypo, draw=none, shape=circle, tikzit fill={rgb,255: red,0; green,123; blue,255}, minimum size=4pt, inner sep=0pt]
\tikzstyle{target}=[fill=conc, draw=none, shape=circle, tikzit fill={rgb,255: red,236; green,26; blue,26}, minimum size=4pt, inner sep=0pt]
\tikzstyle{westhook}=[fill=white, draw=black, shape=diamond, tikzit shape=circle, anchor=west, inner sep=1pt]
\tikzstyle{westpred}=[fill=none, draw=none, shape=circle, anchor=west]
\tikzstyle{binder}=[fill=black, draw=black, shape=circle, inner sep=0pt, minimum size=3pt]
\tikzstyle{scroll anchor}=[fill={rgb,255: red,255; green,111; blue,0}, draw={rgb,255: red,255; green,111; blue,0}, shape=circle, inner sep=0pt, minimum size=3pt]
\tikzstyle{southbox}=[fill=none, draw=black, shape=rectangle, anchor=south]
\tikzstyle{northbox}=[fill=none, draw=black, shape=rectangle, anchor=north]
\tikzstyle{eastbox}=[fill=none, draw=black, shape=rectangle, anchor=east]
\tikzstyle{westbox}=[fill=none, draw=black, shape=rectangle, anchor=west]
\tikzstyle{box}=[fill=none, draw=black, shape=rectangle]

\tikzstyle{pos area}=[-, fill=pbg, draw=black, tikzit fill=white, tikzit draw=black]
\tikzstyle{neg area}=[-, fill=nbg, draw=black, tikzit fill={rgb,255: red,191; green,191; blue,191}, tikzit draw=black]
\tikzstyle{pistil}=[-, tikzit fill={rgb,255: red,255; green,243; blue,68}, tikzit draw=black, fill=pistil, draw=black]
\tikzstyle{justif}=[draw={rgb,255: red,0; green,210; blue,0}, -{Computer Modern Rightarrow[length=1mm,width=2mm,line width=0.25mm]}]
\tikzstyle{teridentity}=[-, fill=black]
\tikzstyle{wire}=[-]
\tikzstyle{black arrow}=[->]

%% file: notations.tex
\newcommand{\vrefl}[1]{\text{\raisebox{\depth}{\scalebox{-1}[-1]{$#1$}}}}

\NewDocumentCommand{\rnm}{mO{}tr}
  {{\small$
    \IfBooleanT{#3}{\overline}{
      #1
    }^{#2}$}}

\NewDocumentCommand{\R}{oO{}trtd}{
  \expandafter\prftree\expanded{
    \IfBooleanT{#4}{[d]}
    \IfValueT{#1}{[r]{
      \unexpanded{\rnm{#1}}
      [#2]
      \IfBooleanT{#3}{r}
    }}
  }
}

\newcommand{\rsf}[1]{\text{\small$\mathsf{#1}$}}

\newcommand{\nats}{\mathbb{N}}

\DeclareRobustCommand{\sys}[1]{$\mathsf{#1}$}

\newcommand{\da}{\downarrow}
\newcommand{\ua}{\uparrow}

\newcommand{\sep}{\;\!;\!\;}

\newcommand{\defeq}{\triangleq}

\newcommand{\deq}{\coloneq}

\knowledgenewrobustcmd\step[1]{\mathrel{\cmdkl{\rightarrow_{#1}}}}
\knowledgenewrobustcmd\steps[1]{\mathrel{\cmdkl{\rightarrow^{\ast}_{#1}}}}
\newrobustcmd\notstep[1]{\withkl{\kl[\step]}{\mathrel{\cmdkl{\nrightarrow_{#1}}}}}
\newrobustcmd\notsteps[1]{\withkl{\kl[\steps]}{\mathrel{\cmdkl{\nrightarrow^{\ast}_{#1}}}}}
\knowledgenewrobustcmd\nsteps[2]{\mathrel{\cmdkl{\rightarrow^{#1}_{#2}}}}
\knowledgenewrobustcmd\lstep{\mathrel{\cmdkl{\rightharpoonup_{\Nature}}}}
\newcommand{\xstep}[1]{\xrightarrow{#1}}

\knowledgenewrobustcmd\bv{\cmdkl{\mathsf{bv}}}
\knowledgenewrobustcmd\fv{\cmdkl{\mathsf{fv}}}
\renewcommand{\emptyset}{\varnothing}
\newcommand{\subst}[3]{#1 \{ #2 / #3 \}}
\newcommand{\compr}[2]{\left\{#1 ~\middle|~ #2\right\}}

\knowledgenewrobustcmd\fset[3]{\cmdkl{\{}#3\cmdkl{\}_{#1}^{#2}}}

\knowledgenewrobustcmd\hole{\cmdkl{\square}}

\knowledgenewrobustcmd\cfill[2]{#1\cmdkl{\{}#2\cmdkl{\}}}
\knowledgenewrobustcmd{\inv}{\cmdkl{\mathsf{inv}}}
\knowledgenewrobustcmd{\binv}{\cmdkl{\mathsf{inv}}}

\newcommand\ctx[1]{{\hat{#1}}}
\newcommand{\cphi}{\ctx{\phi}}
\newcommand{\cpsi}{\ctx{\psi}}
\newcommand{\cxi}{\ctx{\xi}}
\newcommand{\cPhi}{\ctx{\Phi}}
\newcommand{\cPsi}{\ctx{\Psi}}
\newcommand{\cXi}{\ctx{\Xi}}
\newcommand{\cPhiP}{\ctx{\Phi}^{+}}
\newcommand{\cXiP}{\ctx{\Xi}^{+}}
\newcommand{\cPhiN}{\ctx{\Phi}^{-}}
\newcommand{\cXiN}{\ctx{\Xi}^{-}}
\newcommand{\cg}{\ctx{g}}
\newcommand{\ch}{\ctx{h}}
\newcommand{\ck}{\ctx{k}}
\newcommand{\cG}{\ctx{G}}
\newcommand{\cH}{\ctx{H}}
\newcommand{\cK}{\ctx{K}}

\newcommand{\cKP}{\ctx{K}^{+}}

\newcommand{\cKN}{\ctx{K}^{-}}

\newcommand{\back}{\smalltriangleright}
\newcommand{\forw}{\mdwhtcircle}

\newcommand{\HOne}{{\color{hypnum}\ding{172}}}
\newcommand{\HTwo}{{\color{hypnum}\ding{173}}}
\newcommand{\HThree}{{\color{hypnum}\ding{174}}}
\newcommand{\HFour}{{\color{hypnum}\ding{175}}}
\newcommand{\HFive}{{\color{hypnum}\ding{176}}}
\newcommand{\Hyp}[1]{\colorbox{hyp}{#1}}

\knowledgenewrobustcmd\True{\cmdkl{\top}}
\knowledgenewrobustcmd\False{\cmdkl{\bot}}
\knowledgenewrobustcmd\Neg{\mathop{\cmdkl{\neg}}}
\knowledgenewrobustcmd\Conj{\mathbin{\cmdkl{\land}}}
\knowledgenewrobustcmd\Disj{\mathbin{\cmdkl{\lor}}}
\knowledgenewrobustcmd\Impl{\mathbin{\cmdkl{\supset}}}

\knowledgenewrobustcmd\SA{\cmdkl{\textsc{\small SA}}}

\newcommand*\pcut[1]{\tikz[baseline=(char.base)]{
            \node[shape=circle,draw,inner sep=2pt] (char) {$#1$};}}

\knowledgenewrobustcmd\bnodes{\cmdkl{\boldsymbol{\mathrm{N}}_{\boldsymbol{\beta}}}}
\knowledgenewrobustcmd\bgardens{\cmdkl{\boldsymbol{\Gamma}_{\boldsymbol{\beta}}}}
\knowledgenewrobustcmd\bgraphs{\cmdkl{\boldsymbol{\mathrm{G}}_{\boldsymbol{\beta}}}}

\knowledgenewrobustcmd\bcut[1]{\cmdkl{[}#1\cmdkl{]}}
\knowledgenewrobustcmd\bgarden[2]{#1 \mathbin{\cmdkl{\cdot}} #2}

\knowledgenewrobustcmd\entails[1]{\mathrel{\cmdkl{\vdash_{#1}}}}
\newrobustcmd\nentails[1]{\withkl{\kl[\entails]}{\mathrel{\cmdkl{\nvdash_{#1}}}}}

\knowledgenewrobustcmd\Nature{\cmdkl{\text{\ding{96}}}}
\knowledgenewrobustcmd\Culture{\cmdkl{\text{\ding{34}}}}

\knowledgenewrobustcmd\flower[2]{#1 \mathrel{\cmdkl{\csup}} #2}
\knowledgenewrobustcmd\garden[2]{#1 \mathbin{\cmdkl{\cdot}} #2}

\knowledgenewrobustcmd\flowers{\cmdkl{\mathbb{F}}}
\knowledgenewrobustcmd\gardens{\cmdkl{\mathbb{G}}}

\knowledgenewrobustcmd\interp[1]{\cmdkl{\lBrack}#1\cmdkl{\rBrack}}
\newcommand{\toform}[1]{\lfloor #1 \rfloor}

\knowledgenewrobustcmd\chyp[2]{#1 \mathrel{\cmdkl{\succ}} #2}

\newcommand{\pissheet}[1]{\colorbox{pistil}{$#1$}}

\newcommand{\bx}{\mathbf{x}}
\newcommand{\by}{\mathbf{y}}
\newcommand{\bz}{\mathbf{z}}

\knowledgenewrobustcmd\fdepth[1]{\cmdkl{|}#1\cmdkl{|}}
\knowledgenewrobustcmd\cdepth[1]{\cmdkl{|}#1\cmdkl{|}}

\knowledgenewrobustcmd\vars{\cmdkl{\mathscr{V}}}
\knowledgenewrobustcmd\psymbs{\cmdkl{\mathscr{P}}}
\knowledgenewrobustcmd\arity{\cmdkl{\mathsf{ar}}}
\knowledgenewrobustcmd\dom{\cmdkl{\mathsf{supp}}}

\knowledgenewrobustcmd\upd[1]{\mathbin{\cmdkl{|_{#1}}}}
\newcommand{\update}[3]{#1 \upd{#2} #3}
\newcommand{\restr}[2]{#1\!\upd{#2}\!}
\newcommand{\idsubst}{\mathsf{1}}

\newcommand{\tvec}[1]{\vec{\mathbf{#1}}}

\knowledgenewrobustcmd\rewolF[1]{\cmdkl{\text{\ding{95}}}(#1)}

\knowledgenewrobustcmd\kentails[1]{\mathrel{\cmdkl{\vDash_{#1}}}}
\newrobustcmd\notkentails[1]{\withkl{\kl[\kentails]}{\mathrel{\cmdkl{\nvDash_{#1}}}}}
\knowledgenewrobustcmd\kequiv[1]{\mathrel{\cmdkl{\vrefl{{\vDash}} {\vDash}_{#1}}}}

\knowledgenewrobustcmd\evaluation[1]{\cmdkl{$#1$-evaluation}}
\knowledgenewrobustcmd\consistent[1]{\cmdkl{$#1$-consistent}}
\knowledgenewrobustcmd\complete[1]{\cmdkl{$#1$-complete}}
\newrobustcmd\consistency[1]{\withkl{\kl[\consistent]}{\cmdkl{$#1$-consistency}}}
\newrobustcmd\completeness[1]{\withkl{\kl[\complete]}{\cmdkl{$#1$-completeness}}}

\knowledgenewrobustcmd\eforces[3]{#1\mathrel{\cmdkl{\Vdash}}#2\,\cmdkl{[}#3\cmdkl{]}}
\newrobustcmd\neforces[3]{\withkl{\kl[\eforces]}{#1\mathrel{\cmdkl{\nVdash}}#2\,\cmdkl{[}#3\cmdkl{]}}}

\knowledgenewrobustcmd\access{\mathrel{\cmdkl{\leq}}}
\newrobustcmd\revaccess{\withkl{\kl[\access]}{\mathrel{\cmdkl{\geq}}}}

\knowledgenewrobustcmd\completion[1]{\cmdkl{\mathsf{Com}}(#1)}
\knowledgenewrobustcmd\ncompletion[2]{\cmdkl{\mathsf{Com}^{#2}}(#1)}

\newcommand{\tT}{\mathcall{T}}
\newcommand{\tU}{\mathcall{U}}

%% file: eg-disj-imp.tex
\begin{mathpar}
\begin{array}{cc}
  \stkfig{1}{scroll-disj} & \stkfig{1}{scroll-imp} \\
  A \Disj B & A \Impl B
\end{array}
\quad \not= \quad
\begin{array}{cc}
  \stkfig{1}{eg-disj} & \stkfig{1}{eg-imp} \\
  \neg (\neg A \Conj \neg B) & \neg (A \Conj \neg B)
\end{array}
\end{mathpar}

%% file: flower-calculus.tex
\newcommand{\vsp}{\vspace{1em}}
\newcolumntype{C}[1]{>{\centering}m{#1}}

\begin{framed}
\[
\begin{array}{cc}
  \begin{array}{cc@{\vsp}}
  \multicolumn{2}{c@{\vsp}}{\textsc{Nature}~\intro*\Nature}
  \\
  \R[\intro{poll{\da}}]
    {\cfill{\cXi}{}}
    {\cfill{\cXi}{\Phi}}
  &
  \R[\intro{poll{\ua}}]
    {\cfill{\cXi}{\Phi}}
    {\cfill{\cXi}{}}
  \\
  \R[\intro{epis}]
    {\flower{\garden{{}}{{}}}{\garden{{}}{\Phi}}}
    {\Phi}
  &
  \R[\intro{epet}]
    {}
    {\flower{\gamma}{\garden{{}}{{}} \sep \Delta}}
  \\
  \multicolumn{2}{c@{\vsp}}{
    \R[\intro{srep}]
      {\flower{\garden{\bx}{\Phi}}{\garden{{}}{\fset{i}{n}{\flower{\gamma_i}{\Delta}}}}}
      {\flower{\garden{\bx}{\Phi, (\flower{\garden{{}}{{}}}{\fset{i}{n}{\gamma_i}})}}{\Delta}}
  }
  \\
  \multicolumn{2}{c@{\vsp}}{
    \R[\intro{ipis}]
      {(\flower{\garden{\bx}{\sigma(\Phi)}}{\sigma(\Delta)}), (\flower{\garden{\bx, \by}{\Phi}}{\Delta})}
      {\flower{\garden{\bx, \by}{\Phi}}{\Delta}}
  }
  \\
  \multicolumn{2}{c@{\vsp}}{
    \R[\intro{ipet}]
      {\flower{\gamma}{\garden{\bx}{\sigma(\Phi)} \sep \garden{\bx, \by}{\Phi} \sep \Delta}}
      {\flower{\gamma}{\garden{\bx, \by}{\Phi} \sep \Delta}}
  }
  \end{array}
  &
  \hspace{-1.75em}
  \begin{array}{cc@{\vsp}}
  \multicolumn{2}{c@{\vsp}}{\textsc{Culture}~\intro*\Culture}
  \\
  \R[\intro{grow}]
    {\cfill{\cXiP}{\Phi}}
    {\cfill{\cXiP}{}}
  &
  \R[\intro{crop}]
    {\cfill{\cXiN}{}}
    {\cfill{\cXiN}{\Phi}}
  \\
  \R[\intro{pull}]
    {\cfill{\cXiP}{\flower{\gamma}{\Delta}}}
    {\cfill{\cXiP}{\flower{\gamma}{\Gamma \sep \Delta}}}
  &
  \R[\intro{glue}]
    {\cfill{\cXiN}{\flower{\gamma}{\Gamma \sep \Delta}}}
    {\cfill{\cXiN}{\flower{\gamma}{\Delta}}}
  \\
  \multicolumn{2}{c@{\vsp}}{
    \R[\intro{apis}]
      {\cfill{\cXiP}{\flower{\garden{\bx, \by}{\Phi}}{\Delta}}}
      {\cfill{\cXiP}{\flower{\garden{\bx}{\sigma(\Phi)}}{\sigma(\Delta)}}}
  }
  \\
  \multicolumn{2}{c@{\vsp}}{
    \R[\intro{apet}]
      {\cfill{\cXiN}{\flower{\gamma}{\garden{\bx, \by}{\Phi} \sep \Delta}}}
      {\cfill{\cXiN}{\flower{\gamma}{\garden{\bx}{\sigma(\Phi)} \sep \Delta}}}
  }
  \end{array}
\end{array}
\]

\fontsize{9}{9}\selectfont
\centering

In the rules \kl{poll{\da}} and \kl{poll{\ua}}, we assume that
$\chyp{\Phi}{\cXi}$.
\\[0.5em]
In the rules \kl{ipis}, \kl{apis} (resp. \kl{ipet}, \kl{apet}), we
assume some \kl{substitution} $\sigma : \by$ that is \kl{capture-avoiding} in
$\flower{\garden{{}}{\Phi}}{\Delta}$ (resp. $\Phi$).

\end{framed}

%% file: natural-graphical.tex
\begin{framed}
% {\textsc{Nature} $\Nature$}
% \vspace{1.5em}
\begin{mathpar}
  \begin{array}{rcl}
    \input{0.8.tikz}{1}{generic-flower} &
    \xstep{\kl{poll{\da}}} &
  \end{array}
  \and
  \begin{array}{rcl}
    & \xstep{\kl{poll{\ua}}}
    & \input{0.8.tikz}{1}{generic-flower}
  \end{array}
  \\
  \begin{array}{rcl}
    \Phi &
    ~\xstep{\kl{epis}} &
    \input{1.tikz}{1}{epis}
  \end{array}
  \and
  \begin{array}{rcl}
    \input{0.8.tikz}{1}{epet} &
    \xstep{\kl{epet}} &
  \end{array}
  \\
  \begin{array}{rcl}
    \input{0.8.tikz}{1}{ipis-1}
    & \xstep{\kl{ipis}}
    & \input{0.8.tikz}{1}{ipis-2}
  \end{array}
  \and
  \begin{array}{rcl}
    \input{0.8.tikz}{1}{ipet-1}
    & \xstep{\kl{ipet}}
    & \!\!\input{0.8.tikz}{1}{ipet-2}
  \end{array}
  \\
  \begin{array}{rcl}
    \input{0.7.tikz}{1}{srep-1} &
    \xstep{\kl{srep}} &
    \!\!\!\!\!\!\!\!\input{0.7.tikz}{1}{srep-2}
  \end{array}
\end{mathpar}

% \vspace{3em}

% {\textsc{Culture} $\Culture$}
% \vspace{1.5em}
% \begin{mathpar}
% \end{mathpar}

% \vspace{3em}

% In the rules \kl{poll{\da}} and \kl{poll{\ua}}, we assume that
% $\chyp{\Phi}{\Xi\hole}$.

% In the rules \kl{ipis}, \kl{ipet}, \kl{apis}, \kl{apet}, we assume
% some substitution $\sigma : \by \to \terms$.
\end{framed}

%% file: cultural-graphical.tex
\begin{framed}
% {\textsc{Culture} $\Culture$}
% \vspace{1.5em}
\begin{mathpar}
  \begin{array}{rcl@{\vspace{1em}}}
    \phantom{\input{0.8.tikz}{1}{generic-flower}}
    & \xstep{\kl{grow}}
    & \input{0.8.tikz}{1}{generic-flower}
    \\
    \pissheet{\input{0.8.tikz}{1}{generic-flower-neg}}
    & \xstep{\kl{crop}}
    & \pissheet{\phantom{\input{0.8.tikz}{1}{generic-flower-neg}}}
    \\
    \input{0.8.tikz}{1}{pull-1} &
    \xstep{\kl{pull}} &
    \input{0.8.tikz}{1}{pull-2}
    \\
    \pissheet{\input{0.8.tikz}{1}{glue-1}} &
    \xstep{\kl{glue}} &
    \pissheet{\input{0.8.tikz}{1}{glue-2}}
    \\
    \input{0.8.tikz}{1}{apis-1} &
    \xstep{\kl{apis}} &
    \input{0.8.tikz}{1}{apis-2}
    \\
    \pissheet{\input{0.8.tikz}{1}{apet-1}} &
    \xstep{\kl{apet}} &
    \pissheet{\input{0.8.tikz}{1}{apet-2}}
  \end{array}
\end{mathpar}

% \vspace{3em}

% {\textsc{Culture} $\Culture$}
% \vspace{1.5em}
% \begin{mathpar}
% \end{mathpar}

% \vspace{3em}

% In the rules \kl{poll{\da}} and \kl{poll{\ua}}, we assume that
% $\chyp{\Phi}{\Xi\hole}$.

% In the rules \kl{ipis}, \kl{ipet}, \kl{apis}, \kl{apet}, we assume
% some substitution $\sigma : \by \to \terms$.
\end{framed}

%% file: proof-example-textual.tex
\scalebox{0.8}{
\begin{mathpar}
  \begin{array}{ll}
    &\flower{(\flower{\garden{x}{}}{(\flower{p(x)}{})\sep q(x)})}{(\flower{\garden{y}{p(y)}}{\garden{z}{q(z)}})}
    \\\step{\kl{ipet}}
    &\flower{(\flower{\garden{x}{}}{(\flower{p(x)}{})\sep q(x)})}{(\flower{\garden{y}{p(y)}}{q(y)})}
    \\\step{\kl{poll{\ua}}}
    &\flower{(\flower{\garden{x}{}}{(\flower{p(x)}{})\sep q(x)})}{(\flower{\garden{y}{p(y),(\flower{\garden{x}{}}{(\flower{p(x)}{})\sep q(x)})}}{q(y)})}
    \\\step{\kl{ipis}}
    &\flower{(\flower{\garden{x}{}}{(\flower{p(x)}{})\sep q(x)})}{(\flower{\garden{y}{p(y),(\flower{}{(\flower{p(y)}{})\sep q(y)})}}{q(y)})}
    \\\step{\kl{srep}}
    &\flower{(\flower{\garden{x}{}}{(\flower{p(x)}{})\sep q(x)})}{(\flower{\garden{y}{p(y)}}{(\flower{(\flower{p(y)}{})}{q(y)}),(\flower{q(y)}{q(y)})})}
    \\\nsteps{2}{\kl{poll{\da}}}
    &\flower{(\flower{\garden{x}{}}{(\flower{p(x)}{})\sep q(x)})}{(\flower{\garden{y}{p(y)}}{(\flower{(\flower{}{})}{q(y)}),(\flower{q(y)}{\garden{}{}})})}
    \\\step{\kl{srep}}
    &\flower{(\flower{\garden{x}{}}{(\flower{p(x)}{})\sep q(x)})}{(\flower{\garden{y}{p(y)}}{(\flower{}{\garden{}{}}),(\flower{q(y)}{\garden{}{}})})}
    \\\nsteps{2}{\kl{epet}}
    &\flower{(\flower{\garden{x}{}}{(\flower{p(x)}{})\sep q(x)})}{(\flower{\garden{y}{p(y)}}{\garden{}{}})}
    \\\step{\kl{epet}}
    &\flower{(\flower{\garden{x}{}}{(\flower{p(x)}{})\sep q(x)})}{\garden{}{}}
    \\\step{\kl{epet}}
  \end{array}
\end{mathpar}
}

%% file: proof-example-graphical.tex
\scalebox{0.8}{
\newcommand{\tkfig}{\input{0.85.tikz}{0.6}}
\newcommand{\nsp}{\hspace{-0.5em}}
\begin{mathpar}
  \begin{array}{c@{\nsp}c@{\nsp}c@{\nsp}c}
    &\tkfig{proof-example-0}
    &\step{\kl{ipet}}
    &\tkfig{proof-example-1}
    \\
    \step{\kl{poll{\ua}}}
    &\tkfig{proof-example-2}
    &\step{\kl{ipis}}
    &\tkfig{proof-example-3}
    \\
    \step{\kl{srep}}
    &\tkfig{proof-example-4}
    &\nsteps{2}{\kl{poll{\da}}}
    &\tkfig{proof-example-5}
    \\
    \step{\kl{srep}}
    &\tkfig{proof-example-6}
    &\nsteps{2}{\kl{epet}}
    &\tkfig{proof-example-7}
    \\
    \step{\kl{epet}}
    &\tkfig{proof-example-8}
    &\step{\kl{epet}}
    &
  \end{array}
\end{mathpar}
}

%% file: beta.tex
\begin{framed}
% {\textsc{Beta}}
% \vspace{1.5em}
\begin{mathpar}
  \R[\intro{Iter}]
    {\cfill{\cK}{G, \cfill{\cH}{}}}
    {\cfill{\cK}{G, \cfill{\cH}{G}}}
  \and
  \R[\intro{Deit}]
    {\cfill{\cK}{G, \cfill{\cH}{G}}}
    {\cfill{\cK}{G, \cfill{\cH}{}}}
  \\
  \R[\intro{Ins}]
    {\cfill{\cKN}{}}
    {\cfill{\cKN}{G}}
  \and
  \R[\intro{Del}]
    {\cfill{\cKP}{G}}
    {\cfill{\cKP}{}}
  \\
  \R[\intro{Dcut{\da}}]
    {\cfill{\cK}{G}}
    {\cfill{\cK}{\bcut{\bgarden{{}}{[\bgarden{{}}{G}]}}}}
  \and
  \R[\intro{Dcut{\ua}}]
    {\cfill{\cK}{\bcut{\bgarden{{}}{[\bgarden{{}}{G}]}}}}
    {\cfill{\cK}{G}}
  \\
  \R[\intro{Unif{\da}}]
    {\cfill{\cKP}{\bcut{\bgarden{\bx, y}{G}}}}
    {\cfill{\cKP}{\bcut{\bgarden{\bx}{\subst{G}{z}{y}}}}}
  \and
  \R[\intro{Unif{\ua}}]
    {\cfill{\cKN}{\bcut{\bgarden{\bx}{\subst{G}{z}{y}}}}}
    {\cfill{\cKN}{\bcut{\bgarden{\bx, y}{G}}}}
\end{mathpar}
\end{framed}